\documentclass[10pt, twosides]{IEEEtran}
\ifCLASSINFOpdf

\else

\fi

\usepackage{mathrsfs}
\usepackage{amsmath}
\usepackage{amsfonts}
\usepackage{amsmath}
\usepackage{amssymb}
\usepackage{amstext}
\usepackage{graphicx}
\usepackage{indentfirst}
\usepackage{array}
\usepackage{cite}
\usepackage{enumerate}
\usepackage{stmaryrd}
\usepackage[linesnumbered,ruled,vlined]{algorithm2e}
\usepackage{multirow}
\usepackage{multicol}
\usepackage{verbatim}
\usepackage{stfloats}
\usepackage{color}
\usepackage{bm}
\usepackage{multirow}
\usepackage{multicol}

\newtheorem{theorem}{\bf{Theorem}}

\newtheorem{proposition}{\bf{Proposition}}

\newtheorem{definition}{\bf{Definition}}
\newtheorem{remark}{\bf{Remark}}
\newtheorem{example}{\bf{Example}}

\begin{document}
%
\title{Polar Codes: Analysis and Construction Based on Polar Spectrum}

\author{Kai Niu, \IEEEmembership{Member,~IEEE}, Yan Li, \IEEEmembership{Student Member,~IEEE}, Weiling Wu, \IEEEmembership{Member,~IEEE}
\thanks
{
This work is supported by the National Key Research and Development Program of China (No. 2018YFE0205501), National Natural Science Foundation of China (No. 61671080) and Qualcomm Corporation.
K. Niu, Y. Li, W. Wu are with the Key Laboratory of Universal Wireless Communications, Ministry of Education,
Beijing University of Posts and Telecommunications, Beijing, 100876, China (e-mail: \{niukai, Lyan, weilwu\}@bupt.edu.cn). \protect\\
}}

\maketitle

\begin{abstract}
Polar codes are the first class of constructive channel codes achieving the symmetric capacity of the binary-input discrete memoryless channels. But the analysis and construction of polar codes involve the complex iterative-calculation. In this paper, by revisiting the error event of the polarized channel, a new concept, named polar spectrum, is introduced from the weight distribution of polar codes. Thus we establish a systematic framework in term of the polar spectrum to analyze and construct polar codes. By using polar spectrum, we derive the union bound and the union-Bhattacharyya (UB) bound of the error probability of polar codes and the upper/lower bound of the symmetric capacity of the polarized channel. The analysis based on the polar spectrum can intuitively interpret the performance of polar codes under successive cancellation (SC) decoding. Furthermore, we analyze the coding structure of polar codes and design an enumeration algorithm based on the MacWilliams identities to efficiently calculate the polar spectrum. In the end, two construction metrics named UB bound weight (UBW) and simplified UB bound weight (SUBW) respectively, are designed based on the UB bound and the polar spectrum. Not only are these two constructions simple and explicit for the practical polar coding, but they can also generate polar codes with similar (in SC decoding) or superior (in SC list decoding) performance over those based on the traditional methods.
\end{abstract}


\IEEEpeerreviewmaketitle

\section{Introduction}
\label{section_I}
\subsection{Relative Research}
\IEEEPARstart{A}{s} the first constructive capacity-achieving coding scheme, polar codes, invented by Ar{\i}kan in 2009 are a great breakthrough \cite{Polarcode_Arikan} in channel coding theory. Since polar codes demonstrate advantages in error performance and other attractive application prospects, after nine years of research effort, polar codes became a coding standard for the control channel in the fifth generation (5G) wireless communication system \cite{5GNR_38212} in 2018. As the most significant concept, channel polarization was introduced to interpret the behaviour of polar coding. The Bhattacharyya parameter combined with the mutual information chain rule are used to evaluate the reliability of the polarized channel and to analyze the convergence behavior of polar codes. Consequently, this theoretical framework based on the channel polarization plays a key role to design and optimize the polar codes.

According to the dependence of the original channel information, the construction method of polar codes usually falls into two categories. In the first category, by using some channel parameters of the binary-input discrete memoryless channel (B-DMC), such as, erasure probability in binary erasure channel (BEC) or signal-to-noise ratio (SNR) in binary-input additive white Gaussian noise (BI-AWGN) channel, the reliability of the polarized channel can be iteratively calculated or evaluated based on the recursive structure of polar coding and the channel parameters of the original channels are not directly associated with the reliability. In what follows, this category is known as the channel-dependent or implicit construction. In the second category, the reliability of the polarized channel can be ordered based on some channel-independent characteristics of polar codes, which is named as the explicit construction. Due to the independence of channel condition, the explicit construction is more desired for the practical application.

For the implicit construction, Ar{\i}kan first proposed a recursive calculation of Bhattacharyya parameter \cite{Polarcode_Arikan} to evaluate the reliability of the polarized channel, whereas this method is only precise for the coding construction in a BEC. For other B-DMCs, such as binary symmetry channel (BSC) or BI-AWGN channel, exact calculation of Bhattacharyya parameter will involve the high-complexity Monte-carlo integration while approximate iterative calculation will result in some performance losses. Subsequently, Mori \emph{et al.} designed the density evolution (DE) algorithm to track the distribution of logarithmic likelihood ratio (LLR) and calculate the error probability under the successive cancellation (SC) decoding \cite{DE_Mori}, which theoretically exhibits the highest
accuracy. However, the high-precision DE algorithm has a complexity of $O(N\xi \log{\xi})$, where $N$ is the code length, $\xi$ denotes the number of samples and its typical value is about $10^5$, implying it's a time-consuming process in practical application. In \cite{Tal_Vardy}, Tal and Vardy proposed an iterative algorithm to evaluate the upper/lower bound on the error probability of each polarized channel, which can achieve almost the same accuracy as DE with a lower complexity of $O(N{\mu}^2\log\mu)$, where $\mu$ is a fixed integer far less than $\xi$. Afterwards, Trifonov advocated the use of Gaussian approximation (GA) to estimate the error probability in AWGN channel \cite{GA_Trifonov} thereby obtain good accuracy with further reduced complexity $O(N\log N)$. Lately, in order to further increase the accuracy of GA construction, we devised an improved GA algorithm \cite{IGA_Dai} for the polar codes with the long code length. Generally, all these algorithms belong to the implicit construction since the polarized channel reliability is derived based on the iterative calculation which relies on some parameters of the original channel and the large variation of channel condition may affect the reliability order of the polarized channel.

Commonly, from the viewpoint of system design, the coding construction should be explicit and independent of channel condition so as to facilitate the implementation of the encoder and decoder. Therefore, implicit construction is not convenient for the practical application of polar coding, on the contrary, the explicit construction is a more desirable selection. Sch$\rm\ddot{u}$rch \emph{et al.} \cite{PO_Schurch} found that the reliability order of a part of polarized channels is invariant thereby introduced the concept of partial order (PO). It is an efficient tool to reduce the complexity of polar coding construction. Nevertheless, we still need to calculate the reliabilities of the rest polarized channels. Furthermore, He \emph{et al.} \cite{PW_He} exploited the index feature of polarized channels and proposed an explicit construction, named polarized weight (PW) algorithm. Although PW is an empirical construction, amazingly, polar codes constructed by PW can achieve almost the same performance as those constructed by GA algorithm. In fact, the polar codes in 5G standard \cite{5GNR_38212} is constructed by using a fixed polar reliability sequence, which is universal for all the code configuration and obtained by computer searching \cite{5Gpolar_Bioglio}.

\subsection{Motivation}
Traditionally, good channel codes, such as Turbo/LDPC code, can be evaluated and optimized based on the distance spectrum or weight distribution \cite{Book_Lin}. However, the high complexity involved in the weight distribution calculation of polar codes \cite{Distance_spectrum,Weight_distribution} makes it unrealistic to use such metrics for the polar code design. Although there are plenty of research results on the weight distribution in the classic coding theory, it is lack of the theoretic interpretation based on distance spectrum/weight distribution for the channel polarization. Obviously, there is a fault between the research of the classic channel coding and that of polar coding. Thus, in this paper, we focus on the theoretic framework based on distance spectrum to thoroughly interpret the behavior of polar codes and whereby deduce some new constructions that are analytical and explicit.

\subsection{Main Contributions}
In this paper, we introduce a new concept on the distance spectrum of polar codes, named polar spectrum, and establish a complete framework to analyze the performance of the polar codes under the SC decoding. Based on this framework, we obtain new explicit constructions for the polar coding. The main contributions of this paper can be summarized as follows.

\begin{enumerate}[1)]
  \item First, the theoretical framework based on the polar spectrum is established. We revisit the error event analysis of polarized channel and find that one polarized channel is associated with a subcode and its codeword subset, namely polar subcode. Consequently, we set up a one-to-one mapping between the polarized channel and the polar subcode and introduce the weight distribution of the polar subcode, named polar spectrum, which is a set of the weight enumerator with the given Hamming weight. Compared with the Bhattacharyya parameter or other performance metrics needed iterative calculation, the union bound of the error probability of the polarized channel in term of polar spectrum has an intuitive interpretation, that is, the Hamming weight affects the pairwise error probability and the weight enumerator determines the number of the corresponding error events.

      By the polar spectrum, we first build the relevance between the error probability of polarized channel and the weight distribution of polar code, which is a simple analytical-metric rather than a complex iterative-calculation. Furthermore, based on the polar spectrum, the union bound and the union-Bhattacharyya (UB) bound of the block error rate (BLER) under the SC decoding is derived. Meanwhile, the upper/lower bound of the mutual information of the polarized channel is also derived by using the polar spectrum. Therefore, the framework based on polar spectrum can provide the same performance analysis as the recursive calculation based on Bhattacharyya parameter.
  \item Second, thanks to the good structure of (polar) subcodes, we design an iterative enumeration algorithm for the polar spectrum. We find that a pair of specific subcodes can constitute the dual codes. By using the well-known MacWilliams identities \cite{MacWilliams}, the weight distribution of these two subcodes can be easily calculated. Moreover, we prove that the polar weight enumerator of polar subcode and the weight enumerator of subcode satisfy the recursive accumulation relation. On this basis, an iterative algorithm embedded the solution of MacWilliams identities is proposed to enumerate the polar spectrum of polar codes. Compared with the traditional searching algorithm of distance spectrum, such enumeration of polar spectrum is a low complexity and high efficiency algorithm, since its one running can generate all the polar spectra for arbitrary code configuration with a fixed code length.
  \item Third, two explicit and analytical construction metrics, named union-Bhattacharyya weight (UBW) and simplified UB weight (SUBW), are proposed. They are the logarithmic version of the UB bound and the latter, SUBW, only considers the minimum weight term of the UB bound. Since the polar spectrum can be calculated off-line, these two constructions have a linear complexity far below those constructions based on the iterative calculation. On the other hand, they can also be modified as channel-independent constructions for the practical polar coding. Simulation results show that the polar codes constructed by these two metrics can achieve similar performance of those constructed by Tal-Vardy method, GA or PW algorithm under SC decoding. Dramatically, the former can outperform the latter under successive cancellation list (SCL) decoding.
\end{enumerate}

The remainder of the paper is organized as follows. Section \ref{section_II} presents the preliminaries of polar codes, covering polar coding and decoding, error performance analysis and polar code construction. Section \ref{section_III} investigates the error event of polarized channel and introduces the concepts of subcode, polar subcode and polar spectrum. Then, by using polar spectrum, the union bound and union-Bhattacharyya bound of the error probability of the polarized channel are derived and analyzed. In addition, the mutual information of the polarized channel is also analyzed based on the polar spectrum. Furthermore, we explore the subcode duality of polar codes and design an iterative enumeration algorithm to calculate the polar spectrum in Section \ref{section_V}. Numerical analysis for the union bound and UB bound in term of polar spectrum and simulation results for comparing UBW/SUBW construction with the traditional constructions are presented in Section \ref{section_VI}. Finally, Section \ref{section_VII} concludes the paper.

\section{Preliminary of Polar Codes}
\label{section_II}
\subsection{Notation Conventions}
In this paper, calligraphy letters, such as $\mathcal{X}$ and $\mathcal{Y}$, are mainly used to denote sets, and the cardinality of $\mathcal{X}$ is defined as $\left|\mathcal{X}\right|$. The Cartesian product of $\mathcal{X}$ and $\mathcal{Y}$ is written as $\mathcal{X}\times \mathcal{Y}$ and $\mathcal{X}^n$ denotes the $n$-th Cartesian power of $\mathcal{X}$. Especially, the hollow symbol, e.g. $\mathbb{C}$, denotes the codeword set of code or subcode. Let $\llbracket a, b \rrbracket$ denote the continuous integer set $\{a,a+1,\cdots,b\}$.

We write $v_1^N$ to denote an $N$-dimensional vector $\left(v_1,v_2,\cdots,v_N\right)$ and $v_i^j$ to denote a subvector $\left(v_i,v_{i+1},\cdots,v_{j-1},v_j\right)$ of $v_1^N$, $1\leq i,j \leq N$. Occasionally, we use the boldface lowercase letter, e.g. $\mathbf{u}$, to denote a vector. Further, given an index set $\mathcal{A}\subseteq \llbracket 1,N \rrbracket$ and its complement set $\mathcal{A}^c$, we write $v_\mathcal{A}$ and $v_{\mathcal{A}^c}$ to denote two complementary subvectors of $v_1^N$, which consist of $v_i$s with $i\in\mathcal{A}$ or $i\in\mathcal{A}^c$ respectively. Then $1_{\mathcal{A}}$ denotes the indicator function of a set $\mathcal{A}$, that is, $1_\mathcal{A}(x)$ equals $1$ if $x\in A$ and $0$ otherwise.

We use the boldface capital letter, such as $\mathbf{F}_N$, to denote a matrix with dimension $N$. So the notation $\mathbf{F}_N(a:N)$ indicates the submatrix consisting the rows from $a$ to $N$ of the matrix $\mathbf{F}_N$. We use $d_H(\mathbf{u},\mathbf{v})$ to denote the Hamming distance between the binary vector $\mathbf{u}$ and $\mathbf{v}$. Similarly, $w_H(\mathbf{u})$ denotes the Hamming weight of the binary vector $\mathbf{u}$. Given $\forall {\mathbf{a}}, {\mathbf{b}}\in \mathbb{R}^N$, let $\left\|\bf{a}-\bf{b}\right\|$ denote the Euclidian distance between the vector $\bf{a}$ and $\bf{b}$.

Throughout this paper, $\log\left(\cdot\right)$ means ``logarithm to base 2,'' and $\ln\left(\cdot\right)$ stands for the natural logarithm. $(\cdot)^T$ means the transpose operation of the vector. $(\mathbf{a}\cdot \mathbf{b})=\mathbf{a}\mathbf{b}^T$ is the inner product of two vectors $\mathbf{a}$ and $\mathbf{b}$. $\text{supp}(\mathbf{u})=\{i:u_i\neq0\}$ is the set of indices of non-zero elements of a vector $\mathbf{u}$. Let $\lceil x\rceil$ denote the ceiling function, that is, the least integer greater than or equal to $x$.

\subsection{Encoding and Decoding of Polar Codes}
Given a B-DMC $W: \mathcal{X}\to \mathcal{Y}$ with input alphabet $\mathcal{X}= \{0,1\}$ and output alphabet $\mathcal{Y}$, the channel transition probabilities can be defined as $W(y|x)$, $x\in \mathcal{X}$ and $y\in \mathcal{Y}$. If $W$ is symmetric, there is a permutation $\pi_1$ on $\mathcal{Y}$ such that $W(y|1)=W(\pi_1(y)|0)$ for all $y\in \mathcal{Y}$. Let $\pi_0$ be the identity permutation on $\mathcal{Y}$. We write $x\star y$ to denote $\pi_x(y)$, here $x\in \mathcal{X}$ and $y\in \mathcal{Y}$.

Then the symmetric capacity and the Bhattacharyya parameter of the B-DMC $W$ can be defined as
\begin{equation}\label{equation1}
 I(W) \triangleq \sum\limits_{y \in \mathcal{Y}} \sum\limits_{x \in \mathcal{X}} \frac{1}{2} W\left( y\left| x \right. \right)\log \frac{W( y\left| x \right.)}{\tfrac{1}{2}W(y\left| 0 \right.) + \tfrac{1}{2}W(y\left| 1 \right. )}
\end{equation}
and
\begin{equation}\label{equation2}
 Z( W ) \triangleq \sum\limits_{y \in \mathcal{Y}} {\sqrt {W( y|0)W(y|1)}}
\end{equation} respectively.

Applying channel polarization transform for $N=2^n$ independent uses of B-DMC $W$, after channel combining and splitting operation \cite{Polarcode_Arikan}, we obtain a group of polarized channels $W_N^{(i)}: \mathcal{X}\to \mathcal{Y} \times \mathcal{X}^{i-1}$, $i\in \llbracket 1,N \rrbracket$. By using the channel polarization, the polar coding can be described as follows.

Given the code length $N$, the information length $K$ and code rate $R=K/N$, the indices set of polarized channels can be divided into two subsets: one set $\mathcal{A}$, named information set, to carry information bits and the other complement set $\mathcal{A}^c$ to assign the fixed binary sequence, named frozen bits. A message block of $K=|\mathcal{A}|$ bits is transmitted over the $K$ most reliable channels $W_N^{(i)}$ with indices $i\in \mathcal{A}$ and the others are used to transmit the frozen bits.  So a binary source block $u_1^N$ consisting of $K$ information bits and $N-K$ frozen bits can be encoded into a codeword $x_1^N$ by
\begin{equation}\label{equation3}
x_1^N=u_1^N{\bf{F}}_N,
\end{equation}
where the matrix ${{\bf{F}}_N}$ is the $N$-dimension generator matrix. \footnote{In the seminal paper \cite{Polarcode_Arikan}, the generator matrix is composed of the matrix ${\bf{F}}_N$ and the bit-reversal matrix. Since the bit-reversal operation does not affect the reliability of the polarized channel, in this paper, we will use the matrix $\mathbf{F}_N$ as the generator matrix. This is also the polar coding form in 5G standard \cite{5GNR_38212}.} This matrix can be recursively defined as ${{\bf{F}}_N} = {\bf{F}}_2^{ \otimes n}$, where ``$^{\otimes n}$'' denotes the $n$-th Kronecker product and ${{\bf{F}}_2} = \left[ { \begin{smallmatrix} 1 & 0 \\ 1 &  1 \end{smallmatrix} } \right]$ is the $2\times2$ kernel matrix.

Such polar coding can deduce the product-form channel $W^N$ and the synthetic channel $W_N$. The transition probabilities of these two channels satisfy $W_N\left( {y_1^N\left| {u_1^N} \right.} \right) = {W^N}\left( {y_1^N\left| {u_1^N{{\mathbf{F}}_N}} \right.} \right) = \prod\nolimits_{i = 1}^N {W\left( {{y_i}\left| {{x_i}} \right.} \right)}$. So the transition probability of the $i$-th polarized channel $W_N^{(i)}$ is defined as
\begin{equation}\label{equation4}
W_N^{(i)}\left( y_1^N,u_1^{i-1}\left| u_i \right. \right) \triangleq \sum\limits_{u_{i+1}^N \in {\mathcal{X}^{N-i}}} {\frac{1}{2^{N-1}} W_N\left( y_1^N\left| u_1^N \right. \right)}.
\end{equation}
The Bhattacharyya parameter of $W_N^{(i)}$ is expressed as
\begin{equation}\label{equation5}
\begin{aligned}
  &Z\left( {W_N^{\left( i \right)}} \right) \hfill \\
  &= \sum\limits_{y_1^N \in {\mathcal{Y}^N}} {\sum\limits_{u_1^{i - 1} \in {\mathcal{X}^{i - 1}}} {\sqrt {W_N^{\left( i \right)}\left( {y_1^N,u_1^{i - 1}\left| 0 \right.} \right)W_N^{\left( i \right)}\left( {y_1^N,u_1^{i - 1}\left| 1 \right.} \right)} } }.  \hfill \\
\end{aligned}
\end{equation}

As proposed in \cite{Polarcode_Arikan}, polar codes can be decoded by the SC decoding algorithm with a low complexity $O(N\log N)$. Furthermore, many improved SC decoding algorithms, such as successive cancellation list (SCL) \cite{SCL_Tal}, successive cancellation stack (SCS) \cite{SCS_Niu}, successive cancellation hybrid (SCH) \cite{SCH_Chen}, successive cancellation priority (SCP) \cite{SCH_Guan}, and CRC aided (CA)-SCL/SCS \cite{SCL_Tal,CASCL_Niu,ASCL_Li,Survey_Niu} decoding can be applied to improve the performance of polar codes.

\subsection{Error Performance of Polar Codes}
Based on the joint probabilities $P\left( {\left\{ {\left( {u_1^N,y_1^N} \right)} \right\}} \right) = {2^{ - N}}{W_N}\left( {y_1^N\left| {u_1^N} \right.} \right)$, we introduce the probability space $\left(\mathcal{X}^N \times \mathcal{Y}^N,P \right)$, where all $\left( {u_1^N,y_1^N} \right) \in \mathcal{X}^N \times \mathcal{Y}^N$. On this probability space, the random vectors $U_1^N$, $X_1^N$, $Y_1^N$, and $\hat {U}_1^N$ represent the input to the channel $W_N$, the input to the channel $W^N$, the output of $W^N$ or $W_N$, and the decisions by the decoder. Given a sample point $\left( {u_1^N,y_1^N} \right) \in \mathcal{X}^N \times \mathcal{Y}^N$, we have $U_1^N\left(u_1^N,y_1^N\right)=u_1^N$, $X_1^N\left(u_1^N,y_1^N\right)=u_1^N {\bf{F}}_N$, $Y_1^N\left(u_1^N,y_1^N\right)=y_1^N$, and $\hat {U}_1^N\left(u_1^N,y_1^N\right)$ is recursively determined by using SC decoding.

Given the fixed configuration $(N,K,\mathcal{A})$, the block error event under SC decoding is defined as
\begin{equation}
\mathcal{E} \triangleq \left\{ {\left( {u_1^N,y_1^N} \right) \in {\mathcal{X}^N} \times {\mathcal{Y}^N}:{{\hat U}_\mathcal{A}}\left( {u_1^N,y_1^N} \right) \ne {u_\mathcal{A}}} \right\}.
\end{equation}
So the block error rate can be expressed in this probability space as $P_e(N,K,\mathcal{A})=P(\mathcal{E})$. As the derivation in \cite{Polarcode_Arikan}, the set of block error event can be enlarged as $\mathcal{E} \subset \mathop  \bigcup \limits_{i \in \mathcal{A}} {\mathcal{E}_i}$, where the single-bit error event $\mathcal{E}_i$ is defined as
\begin{equation}\label{single_bit_error}
\begin{aligned}
\mathcal{E}_i \triangleq &\left\{\left(u_1^N,y_1^N\right)\in \mathcal{X}^N\times \mathcal{Y}^N: W_N^{(i)}\left(y_1^N,u_1^{i-1}\left|u_i\right.\right)\right.\\
&\left. \leq W_N^{(i)}\left(y_1^N,u_1^{i-1}\left|u_i\oplus 1\right.\right) \right\}.
\end{aligned}
\end{equation}

Thus, the error probability of single-bit error event can be bounded as $P(\mathcal{E}_i)\leq Z\left(W_N^{(i)}\right)$. Further, the BLER under SC decoding can be upper bounded as
\begin{equation}\label{equation8}
P_e(N,K,\mathcal{A})=P(\mathcal{E})\leq \sum \limits_{i\in \mathcal{A}} P(\mathcal{E}_i) \leq \sum \limits_{i\in \mathcal{A}} Z\left(W_N^{(i)}\right).
\end{equation}

\subsection{Polar Code Construction}
For the construction of polar codes, the calculation of channel reliabilities and selection of good channels are the critical steps. Generally, the construction algorithms can be divided into two categories, namely, the channel-dependent (implicit) construction and the explicit construction.

For the former, Ar{\i}kan initially proposed the construction based on the Bhattacharyya parameter \cite{Polarcode_Arikan}, whereas this method is only precise for the BEC and approximate for other channels. Lately, density evolution (DE) algorithm \cite{DE_Mori} and Tal-Vardy algorithm \cite{Tal_Vardy} were proposed to perform high-precise construction of polar codes. However, these algorithms have a slightly high complexity. The Gaussian approximation (GA) algorithm \cite{GA_Trifonov} is a desirable method to construct the polar codes with a medium complexity $O(N\log N)$, especially for the AWGN channel. Further, an improved GA algorithm \cite{IGA_Dai} was designed for the long code length. Nevertheless, these channel-dependent constructions are not convenient for the practical application.

On the contrary, since the explicit construction is independent with the channel condition, it is more desirable for the practical design of polar code. For an example, polarized weight (PW) construction \cite{PW_He} is a typical method in this category. Furthermore, due to the usage of constructive property of the generator matrix, the construction based on partial order \cite{PO_Schurch} is also a good method. However, these methods are still heuristic and do not fully explore the algebraic construction of polar code.

In this paper, we will re-establish the analysis framework of the error performance by investigating the distance spectrum of polar codes so as to obtain analytical construction metrics.
\section{Performance Analysis based on Polar Spectrum}
\label{section_III}
In this section, we will introduce a new analysis tool, namely polar spectrum, to analyze the error performance of polar codes under the SC decoding. First we investigate the error probability of single-bit error event and introduce the concepts of (polar) subcode and polar spectrum. Then we derive the new upper bounds of BLER based on polar spectrum. Further, as an enlarged version of BLER bound, the union-Bhattacharyya bound is also derived. Finally, the upper/lower bound of the mutual information of the polarized channel is analyzed by using this new tool.
\subsection{Error Event Probability Analysis}
Recall that the single-bit error event $\mathcal{E}_i$ is associated with the $i$-th polarized channel, the probability of this error event can be upper bounded by summing those of many codeword error events.
\begin{theorem}\label{theorem1}
Given $\forall u_{i+1}^N \in {\mathcal{X}^{N-i}}$ and suppose the transmission bit is zero, that is, $u_i=0$, then the error probability of the single-bit error event $\mathcal{E}_i$ can be enlarged as
\begin{equation}
\begin{aligned}
P\left( {{{\cal E}_i}} \right) \le \sum\limits_{u_{i + 1}^N\in {{\cal X}^{N - i}}} &P\Big( {W_N}\left( {y_1^N\left| {0_1^N} \right.} \right) \\
& \le {W_N}\left( {y_1^N\left| {0_1^{i - 1},1,u_{i + 1}^N} \right.} \right) \Big).
\end{aligned}
\end{equation}
\end{theorem}
\begin{proof}
By using Proposition 13 in \cite{Polarcode_Arikan}, due to the symmetry of B-DMC $W$, the transition probabilities of synthetic channel $W_N$ and polarized channel $W_N^{(i)}$ can be written as
\begin{equation}\nonumber
\begin{aligned}
&W_N\left(\widetilde{y}_1^N\left|\widetilde{u}_1^{N}\right.\right)\\
&=W_N\left(a_1^N\mathbf{F}_N \star \widetilde{y}_1^N \left|\widetilde{u}_1^{N}\oplus a_1^{N}\right.\right)
\end{aligned}
\end{equation}
and
\begin{equation}\nonumber
\begin{aligned}
&W_N^{(i)}\left(\widetilde{y}_1^N,\widetilde{u}_1^{i-1}\left|\widetilde{u}_i\right.\right)\\
&=W_N^{(i)}\left(a_1^N\mathbf{F}_N \star \widetilde{y}_1^N, \widetilde{u}_1^{i-1}\oplus a_1^{i-1}\left|\widetilde{u}_i\oplus a_i\right.\right)
\end{aligned}
\end{equation} respectively.

Let $a_1^{N}=\widetilde{u}_1^{N}$ and $a_1^N\mathbf{F}_N \star \widetilde{y}_1^N=y_1^N$, we have
\begin{equation}\nonumber
P\left(\mathcal{E}_i\right) =P\left(W_N^{(i)}\left(y_1^N,0_1^{i-1}\left|0\right.\right) \leq W_N^{(i)}\left(y_1^N,0_1^{i-1}\left| 1\right.\right) \right).
\end{equation}

Extending the transition probabilities of polarized channel $W_N^{(i)}$ with Eq. (\ref{equation4}), we have
\begin{equation}\nonumber
\begin{aligned}
P\left(\mathcal{E}_i\right)=
P& \left( \sum\limits_{\widetilde{u}_{i+1}^N \in {\mathcal{X}^{N-i}}} W_N\left( y_1^N\left| 0_1^{i-1},0,\widetilde{u}_{i+1}^N\oplus a_{i+1}^N \right. \right) \right.\\
&\left.\leq \sum\limits_{\widetilde{v}_{i+1}^N \in {\mathcal{X}^{N-i}}} W_N\left( y_1^N\left| 0_1^{i-1},1,\widetilde{v}_{i+1}^N \oplus a_{i+1}^N \right. \right) \right).
\end{aligned}
\end{equation}
Let $\widetilde{v}_{i+1}^N \oplus a_{i+1}^N=u_{i+1}^N$, it follows that
\begin{equation}
\begin{aligned}
P\left(\mathcal{E}_i\right)=
P&\left(\bigcup\limits_{u_{i+1}^N \in {\mathcal{X}^{N-i}}} \Bigg\{ W_N\left( y_1^N\left| 0_1^N \right. \right) \right.\\
&\leq W_N\left( y_1^N\left| 0_1^{i-1},1,u_{i+1}^N \right. \right) \Bigg\} \Bigg).
\end{aligned}
\end{equation}
Using the property of the union set, we complete the proof.
\end{proof}

Concerning the derivation of Theorem \ref{theorem1}, we find that the error event $\mathcal{E}_i$ is associated with the codewords transmitted over the polarized channel $W_N^{(i)}$. Thus we introduce the following definitions.
\begin{definition}\label{definition1}
Given the code length $N$, the $i$-th subcode $\mathbb{C}_N^{(i)}$ is defined as a set of codewords, that is,
\begin{equation}
\mathbb{C}_N^{(i)}\triangleq \left\{{\bf{c}}:{\bf{c}}=\left(0_1^{(i-1)},u_i^N\right){\bf{F}}_N,\forall u_i^N \in \mathcal{X}^{N-i+1} \right\}.
\end{equation}
Furthermore, one subset of the subcode $\mathbb{C}_N^{(i)}$, namely the polar subcode $\mathbb{D}_N^{(i)}$, can be defined as
\begin{equation}
\mathbb{D}_N^{(i)}\triangleq \left\{{\bf{c}}^{(1)}:{\bf{c}}^{(1)}=\left(0_1^{(i-1)},1,u_{i+1}^N\right){\bf{F}}_N,\forall u_{i+1}^N \in \mathcal{X}^{N-i} \right\}.
\end{equation}
The complement set of the polar subcode is defined as $\mathbb{E}_N^{(i)}=\mathbb{C}_N^{(i)}-\mathbb{D}_N^{(i)}\triangleq\left\{{\bf{c}}^{(0)}\right\}$.
\end{definition}

Obviously, subcode $\mathbb{C}_N^{(i)}$ is a linear block code $(N,N-i+1)$ and we have $\left|\mathbb{C}_N^{(i)}\right|=2^{N-i+1}$ and $\left|\mathbb{D}_N^{(i)}\right|=2^{N-i}$. Let us investigate the Hamming distance between the codeword of $\mathbb{D}_N^{(i)}$ and that of $\mathbb{E}_N^{(i)}$.
\begin{proposition}\label{proposition2}
For any ${\bf{c}}^{(1)}\in \mathbb{D}_N^{(i)}$ and ${\bf{c}}^{(0)}\in \mathbb{E}_N^{(i)}$, the Hamming distance between these two codewords satisfies $d_H\left({\bf{c}}^{(1)}, {\bf{c}}^{(0)}\right)=w_H\left({\bf{c}}^{(1)}\oplus {\bf{c}}^{(0)}\right)$ and the module-2 sum of these two codewords is belong to the given polar subcode, that is, ${\bf{c}}^{(1)}\oplus {\bf{c}}^{(0)} \in \mathbb{D}_N^{(i)}$.
\end{proposition}
\begin{proof}
By Definition \ref{definition1}, we have
\begin{equation}\nonumber
\begin{aligned}
&d_H\left({\bf{c}}^{(1)}, {\bf{c}}^{(0)}\right)\\
&=w_H\left(\left(\left(0_1^{(i-1)},1,u_{i+1}^N\right)\oplus \left(0_1^{(i-1)},0,v_{i+1}^N\right)\right) {\bf{F}}_N\right)\\
&=w_H\left(\left(0_1^{(i-1)},1,u_{i+1}^N\oplus v_{i+1}^N\right) {\bf{F}}_N\right).
\end{aligned}
\end{equation} Obviously, ${\bf{c}}^{(1)}\oplus {\bf{c}}^{(0)} \in \mathbb{D}_N^{(i)}$.
\end{proof}

Based on Definition \ref{definition1}, we can further define the pairwise-codeword error event as follows.

\begin{definition}\label{definition2}
Given the codewords ${\bf{c}}^{(0)} \in \mathbb{E}_N^{(i)}$ and ${\bf{c}}^{(1)} \in \mathbb{D}_N^{(i)}$, the pairwise-codeword error event is defined as
\begin{equation}
\begin{aligned}
\mathcal{D}_i\left({\bf{c}}^{(0)}\to {\bf{c}}^{(1)}\right)\triangleq &\left\{\left({\bf{c}}^{(0)},{\bf{c}}^{(1)}, y_1^N\right): W^N\left(y_1^N\left|{\bf{c}}^{(0)}\right.\right)\right. \\
&\left.\leq W^N\left(y_1^N\left|{\bf{c}}^{(1)}\right.\right)\right\}.
\end{aligned}
\end{equation}
\end{definition}
The error probability of the event $\mathcal{D}_i\left({\bf{c}}^{(0)}\to {\bf{c}}^{(1)}\right)$ is named as the pairwise error probability (PEP) and defined as
\begin{equation}\label{equation14}
\begin{aligned}
&P_N^{(i)}\left({\bf{c}}^{(0)}\to {\bf{c}}^{(1)}\right)=P\left( \mathcal{D}_i\left({\bf{c}}^{(0)}\to {\bf{c}}^{(1)}\right) \right)\\
&\triangleq \sum\limits_{{y_1^N}} W^N\left(y_1^N\left|\bf{c}^{(0)}\right.\right) 1_{\mathcal{D}_i\left({\bf{c}}^{(0)}\to {\bf{c}}^{(1)}\right)} \left({\bf{c}}^{(0)}, {\bf{c}}^{(1)},y_1^N\right).
\end{aligned}
\end{equation}

From Proposition \ref{proposition2}, we can simplify the PEP as below,
\begin{equation}\label{equation18}
\begin{aligned}
&P_N^{(i)}\left({\bf{c}}^{(0)}\to {\bf{c}}^{(1)}\right)\\
&=P\left(W^N\left(y_1^N\left|{\bf{c}}^{(0)}\right.\right)  \leq W^N\left(y_1^N\left|{\bf{c}}^{(1)}\right.\right)\right)\\
&=P\left(W^N\left(y_1^N\left|0_1^N \right.\right)  \leq W^N\left(y_1^N\left|{\bf{c}}^{(0)}\oplus{\bf{c}}^{(1)}\right.\right)\right)\\
&=P_N^{(i)}\left(d_H\left({\bf{c}}^{(0)},{\bf{c}}^{(1)}\right)\right).
\end{aligned}
\end{equation}
Without loss of generality, hereafter we designate ${\bf{c}}^{(0)}=0_1^N$. It follows that PEP $P_N^{(i)}\left({\bf{c}}^{(0)}\to {\bf{c}}^{(1)}\right)$ is determined by the codeword weight of polar subcode $\mathbb{D}_N^{(i)}$. 

So we can derive the upper bound of the error probability of the polarized channel $W_N^{(i)}$ as follows.
\begin{theorem}\label{theorem2}
The error probability of $W_N^{(i)}$ is upper bounded by
\begin{equation}
P\left(W_N^{(i)}\right)\leq  \sum\limits_{{\bf{c}}^{(1)}} P_N^{(i)}\left(d_H\left({\bf{c}}^{(0)},{\bf{c}}^{(1)}\right)\right).
\end{equation}
\end{theorem}
\begin{proof}
According to Theorem \ref{theorem1} and Definition \ref{definition2}, we have
\begin{equation}\label{equation17}
\begin{aligned}
P\left(W_N^{(i)}\right)&=P\left(\mathcal{E}_i\right) 
                                \leq P\left(\bigcup \limits_{{\bf{c}}^{(1)}} \mathcal{D}_i\left({\bf{c}}^{(0)}\to {\bf{c}}^{(1)}\right)\right)\\
                                &= \sum \limits_{{\bf{c}}^{(1)}} P\left( \mathcal{D}_i\left({\bf{c}}^{(0)}\to {\bf{c}}^{(1)}\right) \right).
\end{aligned}\\
\end{equation}
Substituting Eq. (\ref{equation18}) into Eq. (\ref{equation17}), we complete the proof.
\end{proof}
Hereafter, $P_N^{(i)}\left(d_H\left({\bf{c}}^{(0)},{\bf{c}}^{(1)}\right)\right)$ is shortly written as $P_N^{(i)}(d)$.

From the above analysis, we establish the 1-1 mapping among the single-bit error event, the polarized channel, and the subcode, that is, $\left\{\mathcal{E}_i \leftrightarrow W_N^{(i)} \leftrightarrow {\mathbb{C}}_N^{(i)}\right\}$. Next we will further analyze the block error rate based on the distance spectrum of polar subcode.
\subsection{Block Error Rate Bound Based on Polar Spectrum}

\begin{definition}\label{definition3}
The polar spectrum of the polar subcode $\mathbb{D}_N^{(i)}$, also named as polar weight distribution, is defined as the weight distribution set $\left\{A_N^{(i)}(d)\right\},d\in \llbracket 1, N \rrbracket $, where $d$ is the Hamming weight of non-zero codeword and the polar weight enumerator $A_N^{(i)}(d)$ enumerates the codewords of weight $d$ for codebook $\mathbb{D}_N^{(i)}$.
\end{definition}

\begin{proposition}\label{proposition3}
The error probability of $W_N^{(i)}$ is further upper bounded by
\begin{equation}\label{equation19}
P\left(W_N^{(i)}\right)\leq \sum \limits_{d=1}^N A_N^{(i)}(d) P_N^{(i)}(d).
\end{equation}
\end{proposition}
\begin{proof}
Collecting the codewords with the same weight of the polar subcode $\mathbb{D}_N^{(i)}$ and using Definition \ref{definition3}, we complete the proof.
\end{proof}

Essentially, this bound is the union bound over the polar subcode. So we further establish the 1-1 mapping among the single-bit error event, the polarized channel, and the polar subcode, that is, $\left\{\mathcal{E}_i \leftrightarrow W_N^{(i)} \leftrightarrow {\mathbb{D}}_N^{(i)}\right\}$. Let $d_{min}^{(i)}$ denote the minimum Hamming distance of polar subcode ${\mathbb{D}}_N^{(i)}$. Thus, by using the PEP formula (\ref{equation18}), we obtain a union bound of BLER under SC decoding.
\begin{theorem}\label{theorem3}
Given the fixed configuration $(N,K,\mathcal{A})$, the block error probability of polar code is upper bounded by
\begin{equation}\label{union_bound}
P_e(N,K,\mathcal{A})\leq \sum\limits_{i\in \mathcal{A}} \sum \limits_{d=1}^N A_N^{(i)}(d) P_N^{(i)}(d).
\end{equation}
\end{theorem}

Compared with the upper bound (\ref{equation8}) proposed by Ar{\i}kan, this union bound has an analytical form, which is mainly determined by the polar spectrum of the selected polar subcodes, that is, the polar weight enumerators and the PEPs dominated by Hamming weights. So this upper bound can reveal more constructive features of polar codes than the traditional bounds. More tighter upper bounds, such as tangential bound or tangential-sphere bound (See \cite{Sason} and references therein), can also be used to evaluate the error performance of polar codes under SC decoding. Nevertheless, these improved upper bounds involve complex calculation and are not convenient for the practical application of polar coding. Therefore, in this paper, we focus on the simple upper bounds, such as union bound and union-Bhattacharyya bound. Next, we will further discuss the PEP under various B-DMC channels, such as BEC, BSC and BI-AWGN channel.

\subsubsection{PEP and BLER in the BEC}
Given the BEC $W:\mathcal{X}\to \mathcal{Y}$, $\mathcal{X}=\{0,1\}$ and $\mathcal{Y}=\{0,e,1\}$ (Here, $e$ denotes an erasure), with the transition probabilities $W(y|x)$ and the erasure probability $\epsilon$, we have $W(e|0)=W(e|1)=\epsilon$.

Let $\mathcal{F}=\{i:y_i=e\}$ denote the set of erasure-occurred indices and assume $|\mathcal{F}|=l$. Assuming ${\bf{c}}^{(0)}=0_1^N$ and $d=d_H({\bf{c}}^{(0)},{\bf{c}}^{(1)})$, we conclude that only the case of the set of erasure-occurred indices covering the support set of ${\bf{c}}^{(0)}+{\bf{c}}^{(1)}$ may result in an error, that is, $l\geq d$ and $\text{supp}\left\{{\bf{c}}^{(0)}+{\bf{c}}^{(1)}\right\}\subseteq\mathcal{F}$. Hence, the PEP in the BEC can be expressed as
\begin{equation}
\begin{aligned}
P_{BEC}\left({\bf{c}}^{(0)}\to{\bf{c}}^{(1)}\right)&=\sum \limits_{l=d}^{N}\left( {\begin{array}{c}  N-d \\ l-d \end{array}} \right)\epsilon^{(l-d)}(1-\epsilon)^{N-l}\epsilon^d\\
&=\epsilon^d.
\end{aligned}
\end{equation}
Furthermore, by Proposition \ref{proposition3}, the union bound of $W_N^{(i)}$ is derived as
\begin{equation}\label{equation22}
P_{BEC}\left(W_N^{(i)}\right)\leq \sum \limits_{d=d_{min}^{(i)}}^N A_N^{(i)}(d)\epsilon^d,
\end{equation}
where $d_{min}^{(i)}$ is the minimum Hamming distance of polar subcode $\mathbb{D}_N^{(i)}$.

Correspondingly, by Theorem \ref{theorem3}, the upper bound of BLER using SC decoding in the BEC can be written as
\begin{equation}\label{equation23}
P_{e,BEC}(N,K,\mathcal{A})\leq \sum\limits_{i\in \mathcal{A}} \sum \limits_{d=d_{min}^{(i)}}^N A_N^{(i)}(d) \epsilon^d.
\end{equation}
\subsubsection{PEP and BLER in the BSC}
Suppose the BSC $W:\mathcal{X}\to \mathcal{Y}$, $\mathcal{X}=\{0,1\}$ and $\mathcal{Y}=\{0,1\}$, with the transition probabilities $W(y|x)$ and the crossover error probability $\delta$, we have $W(1|0)=W(0|1)=\delta$.

Given the transmission codeword ${\bf{c}}^{(0)}=0_1^N$, the received vector can be written as $y_1^N={\bf{c}}^{(0)}+e_1^N$, where $e_1^N$ is the error vector. So the PEP is derived as below.
\begin{theorem}\label{theorem4}
Assuming the decision vector is ${\bf{c}}^{(1)}$ and $d=d_H({\bf{c}}^{(0)},{\bf{c}}^{(1)})$, the PEP in the BSC can be expressed as
\begin{equation}
P_{BSC}\left({\bf{c}}^{(0)}\to{\bf{c}}^{(1)}\right)
= \sum\limits_{m=\lceil d/2 \rceil}^{d}
\left( {\begin{array}{c} d \\ m\end{array}} \right) \delta^m(1-\delta)^{d-m}.
\end{equation}
\end{theorem}
\begin{proof}
When the pairwise error occurs, we have $d_H\left(y_1^N,{\bf{c}}^{(0)}\right)>d_H\left(y_1^N,{\bf{c}}^{(1)}\right)$. This condition is equal to $w_H\left(y_1^N+{\bf{c}}^{(0)}\right)>w_H\left(y_1^N+{\bf{c}}^{(1)}\right)$. Then we have $w_H\left(e_1^N\right)>w_H\left(e_1^N+{\bf{c}}^{(0)}+{\bf{c}}^{(1)}\right)=d_H\left(e_1^N,{\bf{c}}^{(0)}+{\bf{c}}^{(1)}\right)$. It follows that an error will occur if more than half of the elements in the support set of the error vector $e_1^N$ overlap the support set of the codeword ${\bf{c}}^{(0)}+{\bf{c}}^{(1)}$, that is, suppose $\text{supp}(e_1^N)=\mathcal{G}_1$ and $\left|\mathcal{G}_1\right|=m \geq \lceil d/2\rceil$, we have $\mathcal{G}_1\subseteq\text{supp}\left({\bf{c}}^{(0)}+{\bf{c}}^{(1)}\right)$.
Thus we can enumerate the number of set partition of $\mathcal{G}_1$ and complete the proof.
\end{proof}

Furthermore, by Theorem \ref{theorem4}, the union bound of $W_N^{(i)}$ is derived as
\begin{equation}
P_{BSC}\left(W_N^{(i)}\right)\leq \sum \limits_{d=d_{min}^{(i)}}^N \sum\limits_{m=\lceil d/2 \rceil}^{d}\left( {\begin{array}{c} d \\ m\end{array}} \right) \delta^m(1-\delta)^{d-m}.
\end{equation}
Similarly, by Theorem \ref{theorem3}, the upper bound of BLER using SC decoding in the BSC can be written as
\begin{equation}\label{equation26}
P_{e,BSC}(N,K,\mathcal{A})\leq \sum\limits_{i\in \mathcal{A}} \sum \limits_{d=d_{min}^{(i)}}^N \sum\limits_{m=\lceil d/2 \rceil}^{d}\left( {\begin{array}{c} d \\ m\end{array}} \right) \delta^m(1-\delta)^{d-m}.
\end{equation}
\subsubsection{PEP and BLER in the AWGN channel}
For a binary-input AWGN channel, the received signal is expressed as
\begin{equation}
y_j=s_j+n_j,
\end{equation}
where $s_j\in \{\pm \sqrt{E_s}\}$ is the BPSK signal, $E_s$ is the signal energy, and $n_j\sim\mathcal{N}(0,\frac{N_0}{2})$ is a Gaussian noise sample with the zero mean and the variance $\frac{N_0}{2}$.
Since the BPSK modulation is used, the codeword ${\bf{c}}^{(0)}$ is transformed into the transmitted signal vector ${\bf{s}}^{(0)}=\sqrt{E_s}\left({\bf{1}}-2{\bf{c}}^{(0)}\right)$. Thus, the received signal vector can be addressed as
\begin{equation}\label{equation28}
{\bf{y}}=\sqrt{E_s}\left({\bf{1}}-2{\bf{c}}^{(0)}\right)+{\bf{n}},
\end{equation}
where ${\bf{1}}$ is an all-one vector and ${\bf{n}}$ is the AWGN noise vector. Then the transition probability of the product-form channel can be written as
\begin{equation}
W^N\left({\bf{y}\left|{\bf{c}}^{(0)}\right.}\right)=\frac{1}{\left(\pi N_0\right)^{N/2}}\exp\left\{ { - \frac{{{{\left\| {{\mathbf{y}} - {{\mathbf{s}}^{\left( 0 \right)}}} \right\|}^2}}}{{{N_0}}}} \right\}.
\end{equation}
For the PEP of SC decoding in the AWGN channel, we have the following theorem.
\begin{theorem}\label{theorem5}
Assuming the decision vector is ${\bf{c}}^{(1)}$ and $d=d_H\left({\bf{c}}^{(0)},{\bf{c}}^{(1)}\right)$, the PEP between ${\bf{c}}^{(0)}$ and ${\bf{c}}^{(1)}$ can be expressed as
\begin{equation}
P_{AWGN}\left({\bf{c}}^{(0)}\to{\bf{c}}^{(1)}\right)=Q\left[\sqrt{\frac{2E_s}{N_0} d_H\left({\bf{c}}^{(0)},{\bf{c}}^{(1)}\right)}\right],
\end{equation}
where $\frac{E_s}{N_0}$ is the symbol signal-to-noise ratio (SNR) and $Q(x)=\frac{1}{\sqrt{2\pi}}\int_x^{\infty} e^{-t^2/2}dt$ is the tail distribution function of the standard normal distribution.
\end{theorem}
\begin{proof}
First, we assume the codeword ${\bf{c}}^{(0)}$ and ${\bf{c}}^{(1)}$ are mapped to the transmission vector ${\bf{s}}^{(0)}$ and the signal vector ${\bf{s}}^{(1)}$ respectively. If a pairwise error occurs, the Euclidian distances among the received vector and the transmission/signal vectors satisfy the inequality  ${\left\| {\mathbf{y}}- {\mathbf{s}}^{(0)} \right\|^2} > \left\| {\mathbf{y}}-{\mathbf{s}}^{(1)} \right\|^2$. Substituting (\ref{equation28}), we have ${\left\| \mathbf{n}\right\|^2} > \left\|\left({\mathbf{s}}^{(0)}-{\mathbf{s}}^{(1)}\right)+{\mathbf{n}} \right\|^2$. Extending the inequality, we obtain the decision region $\mathcal{H}=\left\{\mathbf{n}:{\mathbf{n}}\left({\mathbf{s}}^{(0)}-{\mathbf{s}}^{(1)}\right)^T<- \frac{1}{2}\left\|{\mathbf{s}}^{(0)}-{\mathbf{s}}^{(1)}\right\|^2\right\}$. So the PEP can be written as
\begin{equation}
\begin{aligned}
  P\left({\mathbf{c}}^{(0)}\to{\mathbf{c}}^{(1)}\right)&= {\int { \cdots \int  } }_{\mathcal{H}} {W_N}\left( {{\mathbf{y}}\left| {{{\mathbf{s}}^{\left( 0 \right)}}} \right.} \right)d{\mathbf{y}} \\
   &= Q\left[ {\sqrt {\frac{{{1}}}{{2{N_0}}}} \left\| {{{\mathbf{s}}^{\left( 0 \right)}} - {{\mathbf{s}}^{\left( 1 \right)}}} \right\|} \right]. \\
\end{aligned}
\end{equation}
Furthermore, due to
\begin{equation}\nonumber
\begin{aligned}
\left\|{\mathbf{s}}^{(0)}-{\mathbf{s}}^{(1)}\right\|&=\left\|\sqrt{E_s}\left[{\mathbf{1}}-2{\mathbf{c}}^{(0)}\right]-\sqrt{E_s}\left[{\mathbf{1}}-2{\mathbf{c}}^{(1)}\right]\right\|\\
&=2\sqrt{E_s}\left\|{\mathbf{c}}^{(1)}-{\mathbf{c}}^{(0)}\right\|,
\end{aligned}
\end{equation}
then $\left\|{\mathbf{s}}^{(0)}-{\mathbf{s}}^{(1)}\right\|=2\sqrt{E_s d_H\left({\mathbf{c}}^{(0)},{\mathbf{c}}^{(1)}\right)}$. So we complete the proof.
\end{proof}

Furthermore, by Theorem \ref{theorem5}, the union bound of $W_N^{(i)}$ is derived as
\begin{equation}
P_{AWGN}\left(W_N^{(i)}\right)\leq \sum \limits_{d=d_{min}^{(i)}}^N A_N^{(i)}(d) Q\left(\sqrt{\frac{2dE_s}{N_0}}\right).
\end{equation}
Similarly, by Theorem \ref{theorem3}, the upper bound of BLER using SC decoding in the BI-AWGN channel can be written as
\begin{equation}\label{equation33}
P_{e,AWGN}(N,K,\mathcal{A})\leq \sum\limits_{i\in \mathcal{A}} \sum \limits_{d=d_{min}^{(i)}}^N A_N^{(i)}(d) Q\left(\sqrt{\frac{2dE_s}{N_0}}\right).
\end{equation}
\subsection{Union-Bhattacharyya Bound}
Union-Bhattacharyya (UB) bound provides a simple form to analyze the error performance bound of the B-DMC channel. Although the UB bound is slightly looser than the union bound, it is convenient to the theoretic analysis.
\begin{proposition}\label{proposition4}
Given the B-DMC $W$, the union-Bhattacharyya (UB) bound of the polarized channel $W_N^{(i)}$ is given by
\begin{equation}\label{equation34}
P\left(W_N^{(i)}\right)\leq \sum \limits_{d=1}^N A_N^{(i)}(d) (Z(W))^d,
\end{equation}
where $Z(W)$ is the Bhattacharyya parameter. Specifically, for the BEC, BSC and AWGN channel, the Bhattacharyya parameter $Z(W)$ is $\epsilon$, $2\sqrt{\delta(1-\delta)}$ and $\exp(-\frac{E_s}{N_0})$, respectively.
\end{proposition}

Furthermore, we can derive the UB bound of BLER as follows.
\begin{theorem}\label{theorem6}
Given the fixed configuration $(N,K,\mathcal{A})$, the UB bound of the block error probability is written as
\begin{equation}\label{UB_bound}
P_e(N,K,\mathcal{A})\leq \sum\limits_{i\in \mathcal{A}} \sum \limits_{d=1}^N A_N^{(i)}(d) (Z(W))^d.
\end{equation}
\end{theorem}
Following we will present the detailed form of UB bound in the BEC, BSC and AWGN channel.

Considering the Bhattacharyya parameter for the three channels, by Proposition \ref{proposition4}, the UB bound of the polarized channel $W_N^{(i)}$ is written as
\begin{equation}
P\left( {W_N^{(i)}} \right) \le \left\{
{\begin{aligned}
&{\sum\limits_{d = d_{min}^{(i)}}^N {A_N^{(i)}} (d){\epsilon^d},{\rm{BEC}}}\\
&{\sum\limits_{d = d_{min}^{(i)}}^N {A_N^{(i)}} (d){{\left( {2\sqrt {\delta (1 - \delta )} } \right)}^d},{\rm{BSC}}}\\
&{\sum\limits_{d = d_{min}^{(i)}}^N {A_N^{(i)}} (d)\exp \left( { - \frac{{d{E_s}}}{{{N_0}}}} \right),{\rm{AWGN}}}
\end{aligned}} \right.
\end{equation}
Note the UB bound is equivalent to the union bound (\ref{equation22}) for the BEC. Similarly, according to Theorem \ref{theorem6}, given the fixed configuration $(N,K,\mathcal{A})$, the UB bound of the BLER using SC decoding in the BEC is equal to Eq. (\ref{equation23}).

Further, the UB bounds of the BLER in the BSC and AWGN channel are respectively given by
\begin{equation}\label{equation38}
P_{e,BSC}(N,K,\mathcal{A})\leq \sum\limits_{i\in \mathcal{A}} \sum \limits_{d=d_{min}^{(i)}}^N A_N^{(i)}(d) {\left(2\sqrt{\delta(1-\delta)}\right)}^d,
\end{equation}
and
\begin{equation}\label{equation40}
P_{e,AWGN}(N,K,\mathcal{A})\leq \sum\limits_{i\in \mathcal{A}} \sum \limits_{d=d_{min}^{(i)}}^N A_N^{(i)}(d) \exp\left(-\frac{dE_s}{N_0}\right).
\end{equation}

\begin{remark}
Observing the union bound and UB bound of the polarized channel, we find that the channel parameters (erasure probability, crossover error probability or symbol SNR) and the weight distribution of the polar subcode explicitly determine the corresponding reliability. On the contrary, the reliability evaluation based on the Bhattacharyya parameter or other metrics needs an iterative calculation and cannot give a simple and direct formula. Since the union/UB bounds reveal the analytical form, we can easily use them to analyze and design the polar codes.
\end{remark}

\subsection{Mutual Information Analysis of Polarized Channel}
Based on the polar spectrum and PEP, we further analyze the symmetric capacity of the polarized channel $W_N^{(i)}$, that is, $I\left(W_N^{(i)}\right)=I\left(U_i;Y_1^N\left|U_1^{i-1}\right.\right)$. At first, we can derive the upper and lower bounds of the Bhattacharyya parameter $Z\left(W_N^{(i)}\right)$ as below.

\begin{proposition}\label{proposition5}
Given the polarized channel $W_N^{(i)}$, the upper bound of the corresponding Bhattacharyya parameter satisfy the following inequality,
\begin{equation}\label{upper_bound_Z}
Z\left(W_N^{(i)}\right) \leq \sum \limits_{d=d_{min}^{(i)}}^N A_N^{(i)}(d) (Z(W))^d.
\end{equation}
\end{proposition}
\begin{proof}
According to the definition in (\ref{equation5}), the Bhattacharyya parameter can be further enlarged as
\begin{equation}
\begin{aligned}
  &Z\left( {W_N^{\left( i \right)}} \right) \\
  &\begin{aligned}
        \leq\frac{1}{2^{N-i}} \sum\limits_{y_1^N }&\sqrt {\sum\limits_{u_{i+1}^N }W_N\left( {y_1^N\left| 0_1^{i-1},0,u_{i+1}^{N} \right.} \right)} \\
            &\cdot\sqrt{\sum\limits_{v_{i+1}^N }W_N\left( {y_1^N\left| 0_1^{i-1},1,v_{i+1}^N \right.} \right)}
      \end{aligned} \\
  &\begin{aligned}
     \leq\frac{1}{2^{N-i}} \sum\limits_{y_1^N }\sum\limits_{u_{i+1}^N }\sum\limits_{v_{i+1}^N }&\sqrt {W_N\left( {y_1^N\left| 0_1^{i-1},0,u_{i+1}^{N} \right.} \right)}\\
      &\cdot\sqrt{ W_N\left( {y_1^N\left| 0_1^{i-1},1,v_{i+1}^N \right.} \right)} .
      \end{aligned}\\
\end{aligned}
\end{equation}
Without loss of generality,  let ${\bf{c}}^{(0)}=\left(0_1^{(i-1)},0,u_{i+1}^N\right){\bf{F}}_N=0_1^N$ and ${\bf{c}}^{(1)}=\left(0_1^{(i-1)},1,v_{i+1}^N\right){\bf{F}}_N$, by using the symmetry of subcode $\mathbb{C}_N^{(i)}$, we have
\begin{equation}\label{equation44}
Z\left( {W_N^{\left( i \right)}} \right)  \leq \sum\limits_{y_1^N \in {\mathcal{Y}^N}}\sum\limits_{{\bf{c}}^{(1)}} {\sqrt {W^N\left( {y_1^N\left| {\bf{c}}^{(0)} \right.} \right)  W^N\left( {y_1^N\left| {\bf{c}}^{(1)} \right.} \right)} }.
\end{equation}
Enumerating all the codewords ${\bf{c}}^{(1)}$ of polar subcode $\mathbb{D}_N^{(i)}$, we obtain the upper bound of Bhattacharyya parameter.
\end{proof}

If we only consider one codeword with the minimum Hamming weight $d_{min}^{(i)}$ in (\ref{upper_bound_Z}), this term can be intended to serve as a lower bound to the Bhattacharyya parameter, that is,
\begin{equation}\label{Z_lower}
\left(Z(W)\right)^{d_{min}^{(i)}}  \lesssim Z\left(W_N^{(i)}\right).
\end{equation}

By using Proposition \ref{proposition5}, we can bound the symmetric capacity of any B-DMC $W$ as the following Theorem.
\begin{theorem}\label{theorem8}
Given the polarized channel $W_N^{(i)}$ and the polar spectrum $\left\{A_N^{(i)}(d)\right\}$ of polar subcode $\mathbb{D}_N^{(i)}$, the symmetric capacity $I\left(W_N^{(i)}\right)=I\left(U_i;Y_1^N\left|U_1^{i-1}\right.\right)$ can be lower bounded by
\begin{equation}\label{I_lower}
I\left(W_N^{(i)}\right)\geq \max\left\{1-\log \left[1+\sum \limits_{d=d_{min}^{(i)}}^N A_N^{(i)}(d) (Z(W))^d\right],0\right\}.
\end{equation}
\end{theorem}
\begin{proof}
According to Proposition 1 in \cite{Polarcode_Arikan}, the symmetric capacity $I\left(W_N^{(i)}\right)$ can be bounded by
\begin{equation}\label{I_upper_lower}
\begin{aligned}
\log \frac{2}{1+Z\left(W_N^{(i)}\right)}\leq I\left(W_N^{(i)}\right)\leq \sqrt{1-Z^2\left(W_N^{(i)}\right)}.
\end{aligned}
\end{equation}
From Proposition \ref{proposition5}, we complete the proof.
\end{proof}

Substituting (\ref{Z_lower}) into (\ref{I_upper_lower}), the symmetric capacity $I\left(W_N^{(i)}\right)$ can be upper bounded by
\begin{equation}\label{I_upper}
I\left(W_N^{(i)}\right)\lesssim \sqrt{1-(Z(W))^{2d_{min}^{(i)}}}.
\end{equation}

Especially, for the BEC, we can further derive the upper/lower bounds of the symmetric capacity of $W_N^{(i)}$ as below.
\begin{proposition}
Given the BEC $W$ with the erasure probability $\epsilon$, then the symmetric capacity of the polarized channel $W_N^{(i)}$ satisfies
\begin{equation}\label{equation46}
I\left(W_N^{(i)}\right)\geq \max\left\{1-\sum \limits_{d=d_{min}^{(i)}}^N A_N^{(i)}(d) \epsilon^d,0\right\}.
\end{equation}
\end{proposition}
\begin{proof}
Noting that $I\left(W_N^{(i)}\right)=1-P\left(W_N^{(i)}\right)\geq 0$ for the BEC and using Eq. (\ref{equation22}), we obtain the lower bound.
\end{proof}

Then by (\ref{I_upper}), it's straightforward to derive the upper bound of $I(W_N^{(i)})$ for the BEC as
\begin{equation}\label{I_upper_BEC}
I\left(W_N^{(i)}\right)\lesssim \sqrt{1-\epsilon^{2d_{min}^{(i)}}}.
\end{equation}

\begin{figure}[h]
\setlength{\abovecaptionskip}{0.cm}
\setlength{\belowcaptionskip}{-0.cm}
  \centering{\includegraphics[scale=0.7]{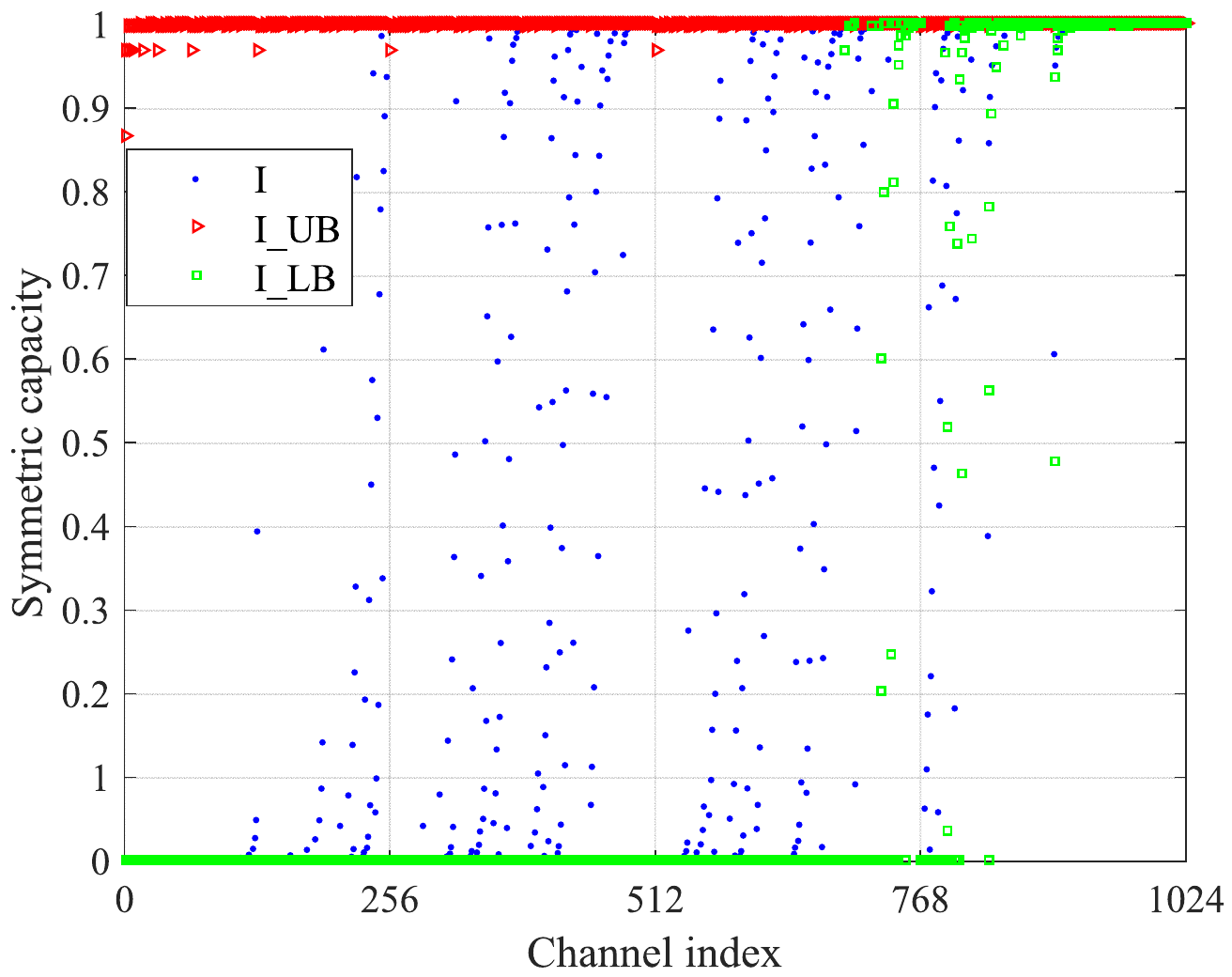}}
  \caption[c]{The upper/lower bounds of the symmetric capacity for a BEC with the erasure probability $\epsilon=0.5$ and the code length $N=1024$.}\label{fig_Channel_polarization}
\end{figure}
Fig. \ref{fig_Channel_polarization} provides the upper/lower bounds of the symmetric capacity of the polarized channels for a BEC with the erasure probability $\epsilon=0.5$ and the code length $N=1024$. The symmetric capacity, marked by ``I'', is evaluated by the recursive relation proposed by Ar{\i}kan (See Eq. (38) in \cite{Polarcode_Arikan}). The upper and lower bounds, marked by ``I\underline{ }UB'' and ``I\underline{ }LB'' respectively, are calculated by (\ref{I_upper_BEC}) and (\ref{equation46}). Although the upper/lower bounds are the coarse approximation of the symmetric capacity, they also reveal the polarization phenomenon.

We establish an analytical framework to evaluate the error performance and the symmetric capacity of polar codes. The polar spectrum is the critical factor to indicate the bounds of the error probability and the mutual information of the polarized channel. So we can use the upper bound of the error probability as a metric to select the information set $\mathcal{A}$.

\section{Calculation of Polar Spectrum}
\label{section_V}
In this section, we will discuss the calculation of polar spectrum. First, we investigate the subcode duality of polar code. Then, by using such duality, we establish the MacWilliams identity between the subcodes. Finally, the enumeration algorithm of the polar spectrum is designed to calculate the polar weight enumerator.
\subsection{Subcode Duality}
Let $i \in \llbracket N/2+1, N\rrbracket$ denote the row index of the matrix $\mathbf{F}_N$. By the definition of subcode $\mathbb{C}_N^{(i)}$, its generator matrix $\mathbf{G}_{\mathbb{C}_N^{(i)}}$ is composed of the rows (from the $i$-th to the $N$-th row) of the matrix $\mathbf{F}_N$, that is, $\mathbf{G}_{\mathbb{C}_N^{(i)}}=\mathbf{F}_N(i:N)$. Further, we introduce another subcode $\mathbb{C}_N^{(N+2-i)}$, whose generator matrix satisfies $\mathbf{G}_{\mathbb{C}_N^{(N+2-i)}}=\mathbf{F}_N(N+2-i:N)$. Thus, subcode $\mathbb{C}_N^{(i)}$ is a linear code $(N,N-i+1)$ and its code rate is $R_{\mathbb{C}_N^{(i)}}=\frac{N-i+1}{N}=1-\frac{i-1}{N}$. Similarly, subcode $\mathbb{C}_N^{(N+2-i)}$ is a linear code $(N,i-1)$ and the code rate is $R_{\mathbb{C}_N^{(N+2-i)}}=\frac{i-1}{N}$.

\begin{theorem}\label{theorem10}
Given $N/2+1\leq i\leq N$, subcode $\mathbb{C}_N^{(N+2-i)}$ is the dual code of subcode $\mathbb{C}_N^{(i)}$, that is, $\mathbb{C}_N^{(N+2-i)}=\mathbb{C}_N^{\bot(i)}$. Especially, $\mathbb{C}_N^{(N/2+1)}$ is a self-dual code.
\end{theorem}
\begin{proof}
We can use mathematical induction to prove this theorem.
Given the code length $N$ and the index $N/2+1\leq i\leq N$, suppose the conclusion $\mathbb{C}_N^{(N+2-i)}=\mathbb{C}_N^{\bot(i)}$ holds. When the code length is doubled to $2N$ and the index is changed to $N+1\leq l\leq 2N$, we obtain two new subcodes $\mathbb{C}_{2N}^{(l)}$ and $\mathbb{C}_{2N}^{(2N+2-l)}$.

By the Plotkin structure $\left[\mathbf{u}+\mathbf{v} \left| \mathbf{v} \right.\right]$ of polar coding \cite{Polarcode_Arikan}, due to $N+1\leq l\leq 2N$, the subcode $\mathbb{C}_{2N}^{(l)}$ consists of two identical component codes $\mathbb{C}_{N}^{(l-N)}$. That is to say, $\forall \mathbf{r}\in \mathbb{C}_{2N}^{(l)}$, we have $\mathbf{r}=(\mathbf{t},\mathbf{t}), \mathbf{t}\in \mathbb{C}_{N}^{(l-N)}$.

On the other hand, by the Poltkin structure, due to $2\leq 2N+2-l\leq N+1$, the subcode $\mathbb{C}_{2N}^{(2N+2-l)}$ is also regarded as a combination of two component codes: one is $\mathbb{C}_N^{(2N+2-l)}$ and the other is $\mathbb{C}_N^{(1)}$. That is, for every $\mathbf{w}\in \mathbb{C}_{2N}^{(2N+2-l)}$, it can be written as $\mathbf{w}=(\mathbf{a}+\mathbf{b},\mathbf{b})$, where $\mathbf{a}\in \mathbb{C}_N^{(2N+2-l)}$ and $\mathbf{b}\in \mathbb{C}_N^{(1)}$.

Now we calculate the inner product of $\mathbf{r}$ and $\mathbf{w}$, that is,
\begin{equation}
\left({\mathbf{r}}\cdot {\mathbf{w}}\right) = \left({\mathbf{t}},{\mathbf{t}}\right)
\left( \begin{array}{*{20}{c}}
  {\mathbf{a}}^T + {\mathbf{b}}^T \\
  {\mathbf{b}}^T
\end{array}\right)
={\mathbf{t}}{\mathbf{a}}^T.
\end{equation}
Due to ${\mathbf{t}}\in \mathbb{C}_N^{(i)}$ and ${\mathbf{a}}\in \mathbb{C}_N^{\bot(i)}$, we have ${\mathbf{t}}{\mathbf{a}}^T=0$. Thus, it follows that $\mathbb{C}_{2N}^{(2N+2-l)}=\mathbb{C}_{2N}^{\bot(l)}$.
\end{proof}

Let $S_N^{(i)}(j)(1\leq j\leq N)$ denote the weight enumerators of the subcode $\mathbb{C}_N^{(i)}$, where $j$ is the Hamming weight of non-zero codeword of codebook $\mathbb{C}_N^{(i)}$. Similarly, $S_N^{\bot( i)}(j)$ denote the weight enumerators of the dual code $\mathbb{C}_N^{\bot (i)}$.
\begin{proposition}\label{proposition8}
Given the subcode $\mathbb{C}_N^{(i)}$, we have $\mathbb{C}_N^{(i)}=\mathbb{D}_N^{(i)}\bigcup \mathbb{C}_N^{(i+1)}$. Thus, the weight enumerator and the polar weight enumerator satisfy $S_N^{(i)}(j)=A_N^{(i)}(j)+S_N^{(i+1)}(j)$.
\end{proposition}
\begin{proof}
Recall that the $i$-th subcode satisfies $\mathbb{C}_N^{(i)}=\mathbb{D}_N^{(i)}\bigcup \mathbb{E}_N^{(i)}$. Since $\mathbb{E}_N^{(i)}=\mathbb{C}_N^{(i+1)}$, then we have $\mathbb{C}_N^{(i)}=\mathbb{D}_N^{(i)}\bigcup \mathbb{C}_N^{(i+1)}$. It follows the relationship between the weight enumerator and the polar weight enumerator.
\end{proof}

\begin{proposition}\label{proposition9}
The odd codeword weight of the subcode $\mathbb{C}_N^{(i)}, i\in \llbracket 2,N \rrbracket$ is zero, that is, $S_N^{(i)}(2j+1)=0$. Similarly, for the polar subcode $\mathbb{D}_N^{(i)},i\in \llbracket 2,N \rrbracket$, we have $A_N^{(i)}(2j+1)=0$. For the polar subcode $\mathbb{D}_N^{(1)}$, the even codeword weight is zero, that is $A_N^{(1)}(2j)=0$.
\end{proposition}
\begin{proof}
Recall that ${{\bf{F}}_N} = {\bf{F}}_2^{ \otimes n}$, then we can obtain the conclusions by the following induction argument due to the Plotkin structure $[\bf{u}+\bf{v}|\bf{v}]$.
\begin{itemize}
  \item [1)] Considering $n=1$, that is ${{\bf{F}}_2} = \left[ { \begin{smallmatrix} 1 & 0 \\ 1 &  1 \end{smallmatrix} } \right]$, it's easy to obtain that $\begin{array}{*{20}{c}}
        j&0&1&2\\
        \hline
        {S_2^{(1)}(j)}&1&2&1\\
        \hline
        {S_2^{(2)}(j)}&1&0&1
        \end{array}$ and
        $\begin{array}{*{20}{c}}
        j&0&1&2\\
        \hline
        {A_2^{(1)}(j)}&0&2&0\\
        \hline
        {A_2^{(2)}(j)}&0&0&1
        \end{array}$. Obviously, the proposition holds under $n=1$.
  \item [2)] Assume the proposition holds under ${{\bf{F}}_N} = {\bf{F}}_2^{ \otimes n}$, where $n$ is an arbitrary integer that satisfies $n\geq2$, 
      then consider the situation of ${{\bf{F}}_{2N}} = {\bf{F}}_2^{ \otimes (n+1)}$. Due to the matrix structure ${\mathbf{F}}_{2N}= \left[ { \begin{smallmatrix} {\mathbf{F}}_{N} & \mathbf{0} \\ {\mathbf{F}}_{N} &  {\mathbf{F}}_{N} \end{smallmatrix} } \right]$, the discussion can be divided into three parts.
  \item [3)] First, in the case of $N+1\leq i\leq2N$, it's straight to know that $S_{2N}^{(i)}(2j)=S_N^{(i-N)}(j)$ and $A_{2N}^{(i)}(2j)=A_N^{(i-N)}(j)$, where $j \in \llbracket 0, N\rrbracket$. Thus the odd weight enumerators of the subcode $\mathbb{C}_{2N}^{(i)}$ and polar subcode $\mathbb{D}_{2N}^{(i)}$ are zero.
  \item [4)] Second, in the case of $2\leq i\leq N$, consider $i=N$ first, for every codeword ${\bf{w}}\in\mathbb{C}_{2N}^{(N)}$, it can be written as ${\bf{w}}={\bf{a}}+{\bf{b}}$, where ${\bf{a}}=(\mathbf{t},\mathbf{0}), {\mathbf{t}}\in {\mathbb{C}}_{N}^{(N)}$ and ${\bf{b}} \in {\mathbb{C}}_{2N}^{(N+1)}$. From 2) and 3), it can be concluded that every ${\bf{w}}\in\mathbb{C}_{2N}^{(N)}$ has the even codeword weight since both ${\bf{a}}$ and ${\bf{b}}$ have the even weight. Next consider $i=N-1$, every codeword ${\bf{w}}\in\mathbb{C}_{2N}^{(N-1)}$ also can be written as ${\bf{w}}={\bf{a}}+{\bf{b}}$, where ${\bf{a}}=(\mathbf{t},\mathbf{0}), {\mathbf{t}}\in {\mathbb{C}}_{N}^{(N-1)}$ and ${\bf{b}} \in {\mathbb{C}}_{2N}^{(N)}$. From 2) and above analysis for $i=N$, it's easy to know every ${\bf{w}}\in\mathbb{C}_{2N}^{(N-1)}$ has the even codeword weight. And so on, same conclusions can be drawn when $i=N-2,N-3,\cdots,2$. By using Proposition \ref{proposition8}, it can be also derived that $A_{2N}^{(i)}(2j+1)=0$.
   \item [5)] Finally, for $i=1$, every codeword ${\bf{w}}\in\mathbb{D}_{2N}^{(1)}$ can be written as ${\bf{w}}={\bf{c}}+{\bf{d}}$, where ${\bf{c}} = \left( {1,0,0, \cdots ,0} \right)$ is a $2N$-length vector with codeword weight one, and ${\bf{d}} \in {\mathbb{C}}_{2N}^{(2)}$ has the even weight from 4), whereby $\bf{w}$ has the odd codeword weight, i.e., $A_{2N}^{(1)}(2j)=0$.
      So, the proposition holds for ${{\bf{F}}_{2N}} = {\bf{F}}_2^{ \otimes (n+1)}$.
\end{itemize}
\end{proof}

\begin{proposition}\label{proposition10}
The weight distribution of the subcode $\mathbb{C}_N^{(i)}$ is symmetric, that is, $S_N^{(i)}(j)=S_N^{(i)}(N-j)$. Similar result holds for the polar subcode $\mathbb{D}_N^{(i)}$, that is, $A_N^{(i)}(j)=A_N^{(i)}(N-j)$.
\end{proposition}
\begin{proof}
Since all-one codeword is always contained in the set $\mathbb{C}_N^{(i)}$, we can easily obtain that $S_N^{(i)}(j)=S_N^{(i)}(N-j)$ by Theorem 1.4.5 (iii) in \cite{Book_Huffman}. Then by Proposition \ref{proposition8}, it can be also derived that $A_N^{(i)}(j)=A_N^{(i)}(N-j)$.

\end{proof}
\subsection{MacWilliams Identities of Subcodes}
It is well known that the linear relations between the weight distributions of a linear code and its dual can be determined by the MacWilliams identities \cite{MacWilliams}. That means, if we know the weight distribution of one linear code, we can obtain the weight distribution of the dual code without knowing the special coding structure. These identities have become the most significant tools to investigate and calculate the weight distributions.

\begin{theorem}\label{theorem11}
Given the subcode $\mathbb{C}_N^{(i)}$ and its dual $\mathbb{C}_N^{\bot(i)}=\mathbb{C}_N^{(N+2-i)}$, the weight enumerators $S_N^{(i)}(j)$ and $S_N^{\bot(i)}(j)$ satisfy the following MacWilliams identities \cite{MacWilliams}
\begin{equation}
\begin{aligned}
&\sum\limits_{j=0}^N \left(\begin{array}{*{20}{c}}N-j\\k\end{array}\right) S_N^{\bot(i)}(j)=2^{i-1-k} \sum\limits_{j=0}^N \left(\begin{array}{*{20}{c}}N-j\\N-k\end{array}\right)S_N^{(i)}{( j)},
\end{aligned}
\end{equation}
where $k \in \llbracket 0,N\rrbracket$.
\end{theorem}

These identities are composed of $N+1$ linear equations. By solving these equations, we can calculate the weight distribution of one subcode and its dual.
\subsection{Enumeration Algorithm of Polar Spectrum}
Using the Plotkin structure $[\bf{u+v}|\bf{v}]$ and the MacWilliams identities, we design an iterative enumeration algorithm to calculate the polar spectrum. Given the weight distribution of subcodes $\mathbb{C}_N^{(i)}$, as shown in Algorithm \ref{algorithm1}, this algorithm enumerates the weight distribution of subcodes $\mathbb{C}_{2N}^{(i)}$ and the polar spectrum of polar subcodes $\mathbb{D}_{2N}^{(i)}$.

\begin{algorithm}[h] \label{algorithm1}
\setlength{\abovecaptionskip}{0.cm}
\setlength{\belowcaptionskip}{-0.cm}
\caption{Iterative enumeration algorithm of polar spectrum}
\KwIn {The weight distribution of all the subcodes with the code length $N$, $\left\{S_N^{(i)}(j):i=1,2,...,N,j=0,1,...,N\right\}$;}
\KwOut {The polar spectrum of all the polar subcodes with the code length $2N$, $\left\{A_{2N}^{(l)}(j):l=1,2,...,2N,j=0,1,...,2N\right\}$;}
\For{$i = 1 \to N-1 $}
{
    \For{$j = 0 \to N $}
     {
           Calculate the polar spectrum of the polar subcodes with the code length $N$, $A_N^{(i)}(j)=S_N^{(i)}(j)-S_N^{(i+1)}(j)$\;
     }
}
\For{$j = 0 \to N $}
{
    Calculate the polar spectrum of the polar subcode which only contains an all-one codeword,
    $A_N^{(N)}(j)= 0$, $j = 0,1,...,N-1$ and $A_N^{(N)}(N)= 1$\;
}
\For{$l = N+1 \to 2N $}
{
    \For{$j = 0 \to {N}$}
    {
        Calculate the polar spectrum of the polar subcodes with the code length $2N$, $A_{2N}^{(l)}(2j)=A_{N}^{(l-N)}(j)$\;
        Calculate the weight distribution of the subcodes with the code length $2N$, $S_{2N}^{(l)}(2j)=S_{N}^{(l-N)}(j)$\;
    }
}
\For{$l = 2 \to N $}
{
    Solve the MacWilliams identities and calculate the weight distribution $S_{2N}^{(l)}(j)$
        \begin{equation}
        \begin{aligned}
         &  \sum\limits_{j=0}^{2N} \left(\begin{array}{*{20}{c}}  2N-j\\ k \end{array}\right) S_{2N}^{(2N+2-l)}(j)\\
         &=2^{l-1-k} \sum\limits_{j=0}^{2N} \left(\begin{array}{*{20}{c}} 2N-j\\2N-k \end{array}\right)S_{2N}^{(l)}{(j)},
        \end{aligned}
        \end{equation}
        where $k=0,1,...,2N$\;
}
\For{$l = 2 \to N $}
{      \For{$j = 0 \to 2N $}
      {
            Calculate the polar spectrum $A_{2N}^{(l)}(j)=S_{2N}^{(l)}(j)-S_{2N}^{(l+1)}(j)$\;
      }
}
\For{$j = 0 \to 2N $}
{
    Initialize the weight distribution $S_{2N}^{(1)}(j)=\left(
             \begin{array}{*{20}{c}}
                      2N\\
                       j
             \end{array}\right)$\;
    Calculate the polar spectrum
    $A_{2N}^{(1)}(j)=S_{2N}^{(1)}(j)-S_{2N}^{(2)}(j)$.
}

\end{algorithm}

Based on the iterative structure, Algorithm \ref{algorithm1} is consist of four steps to enumerate the weight distribution and the polar spectrum. In the first step, by using Proposition \ref{proposition8}, the polar spectrum of polar subcodes $\mathbb{D}_N^{(i)}$ with the code length $N$ is calculated. When the code length is grown from $N$ to $2N$, due to the matrix structure ${\mathbf{F}}_{2N}= \left[ { \begin{smallmatrix} {\mathbf{F}}_{N} & \mathbf{0} \\ {\mathbf{F}}_{N} &  {\mathbf{F}}_{N} \end{smallmatrix} } \right]$, the enumeration can be divided into two parts.

That is to say, for the second step, we enumerate the weight distribution and polar spectrum in the case of $N+1 \leq l\leq 2N$. In this case, by using the Plotkin structure $\left[\mathbf{u}+\mathbf{v} \left| \mathbf{v} \right.\right]$, the subcode $\mathbb{C}_{2N}^{(l)}$ is consist of two identical component codes $\mathbb{C}_{N}^{(l-N)}$. Therefore, the subcode $\mathbb{C}_{2N}^{(l)}$ has the same weight distribution as the subcode $\mathbb{C}_{N}^{(l-N)}$. And similarly, the polar subcode $\mathbb{D}_{2N}^{(l)}$ has the same polar spectrum as the polar subcode $\mathbb{D}_{N}^{(l-N)}$. Certainly, the codeword weight of $\mathbb{C}_{2N}^{(l)}$ or $\mathbb{D}_{2N}^{(l)}$ is doubled.

For the third step, we enumerate the weight distribution and polar spectrum in the case of $1 \leq l\leq N$. According to Theorem \ref{theorem10}, the subcode $\mathbb{C}_{2N}^{(l)}$ is the dual of subcode $\mathbb{C}_{2N}^{(2N+2-l)}$. Since the weight distribution of $\mathbb{C}_{2N}^{(2N+2-l)}$ has been obtained in the second step, we can solve the MacWilliams identities to calculate the weight distribution of $\mathbb{C}_{2N}^{(l)}$.

Finally, using Proposition \ref{proposition8}, the polar spectrum of polar subcodes $\mathbb{D}_{2N}^{(l)}$ with the code length $2N$ is calculated. Note that the weight distribution of subcode $\mathbb{C}_{2N}^{(1)}$ obeys the binomial distribution.

The computational complexity of Algorithm \ref{algorithm1} is mainly determined by the solution of MacWilliams identities. Due to the lower-triangle structure of the coefficient matrix in the MacWilliams identities, the weight enumerators can be recursively calculated. Hence, given the code length $N$, the worse complexity of solving MacWilliams identities is $\chi_M(N)=(N+1)^2$. Furthermore, we only consider to solve $N/2-1$ groups of those identities. So the computational complexity is $\chi_E(N)=(N/2-1)(N+1)^2$. If we consider the non-existence of the odd-weight codewords (Proposition \ref{proposition9}) and the symmetry of the weight distribution (Proposition \ref{proposition10}), the worse complexity of enumeration algorithm can be further reduced to $\chi_E(N)=(N/2-1)(N/4+1)^2$. So the total complexity of
Algorithm \ref{algorithm1} is $O(N^3)$.

\begin{example}
Table \ref{polar_spectrum} shows partial results of the polar spectrum for the code length $N=32$. In this example, due to the symmetric distribution of polar spectrum (Proposition \ref{proposition10}), we only provide half of the polar weight distribution (the number inside the bracket is the symmetric weight), e.g. $A_{32}^{(3)}(2)=A_{32}^{(3)}(30)=128$. As shown in this table, due to the duality relationship, $A_{32}^{(2)}$ and $A_{32}^{(32)}$, $A_{32}^{(3)}$ and $A_{32}^{(31)}$, etc. satisfy the MacWilliams identities. Meanwhile, the subcode $\mathbb{C}_{32}^{(17)}$ is a self-dual code. Furthermore, by Proposition \ref{proposition9}, all the polar weight distributions for the index $i\in \llbracket 2,32 \rrbracket$ have non-zero even codeword weight. On the contrary, for the index $i=1$, non-zero codeword weight is odd.
\end{example}

\begin{table}[tp]
\centering
\caption{Polar spectrum example for $N=32$} \label{polar_spectrum}
\begin{tabular}{|c|c|c|c|c|c|}
\hline index $i$ & weight $d$ & $A_{32}^{(i)}(d)$ & index $i$ & weight $d$ & $A_{32}^{(i)}(d)$ \\
\hline   1 &  1(31) &           32 & 32 & 32 & 1 \\
\hline   1 &  3(29) &        4960 & 31 & 16 & 2 \\
\hline   1 &  5(27) &     201376 & 30 & 16 & 4 \\
\hline   1 &  7(25) &    3365856 & 29 &  8(24) & 4 \\
\hline   1 &  9(23) &   28048800 & 28 & 16 & 16 \\
\hline   1 & 11(21) & 129024480 & 27 & 16 & 16 \\
\hline   1 & 13(19) & 347373600 & 27 &  8(24) & 8 \\
\hline   1 & 15(17) & 565722720 & 26 & 16 & 32 \\
\hline   2 &  2(30) &          256 & 26 &  8(24) & 16 \\
\hline   2 &  4(28) &       17920 & 25 & 12(20) & 56 \\
\hline   2 &  6(26) &     453376 & 25 &   4(28) &  8 \\
\hline   2 &  8(24) &    5258240 & 24 & 16 & 256 \\
\hline   2 & 10(22) &  32258304 & 23 & 16 & 192 \\
\hline   2 & 12(20) & 112892416 & 23 & 12(20) & 128 \\
\hline   2 & 14(18) & 235723520 & 23 &  8(24) &   32 \\
\hline   2 & 16 & 300533760 & 22 & 16 & 384 \\
\hline   3 &  2(30) &          128 & 22 & 12(20) & 256 \\
\hline   3 &  4(28) &        8960 &  22 &  8(24) &  64 \\
\hline   3 &  6(26) &     226688 &  21 & 16 & 832 \\
\hline   3 &  8(24) &    2629120 &  21 & 12(20) & 496 \\
\hline   3 & 10(22) &  16129152 &  21 &   8(24) & 96 \\
\hline   3 & 12(20) &  56446208 &  21 &   4(28) & 16 \\
\hline   3 & 14(18) & 117861760 &  20 & 16 & 1536 \\
\hline   3 & 16 & 150266880 &  20 & 12(20) & 1024 \\
\hline   ... & ... &          ... &  20 &   8(24) & 256 \\
\hline  17 &  2(30) &           16 &  19 &  16 & 3200 \\
\hline  17 &  6(26) &          560 &  19 &  12(20) & 2016 \\
\hline  17 & 10(22) &        4368 &  19 &   8(24) &  448 \\
\hline  17 & 14(18) &       11440 &  19 &   4(28) &   32 \\
\hline
\end{tabular}
\end{table}
\section{Construction of Polar Codes}
\label{section_VI}
In this section, we consider the construction of polar codes based on the UB bound. First, we derive the logarithmic version metric based on UB bound, named UB weight (UBW). Then a more simple metric, named simplified UB weight (SUBW), is designed.
\subsection{Construction Metrics based on UB bound}
By Proposition \ref{proposition4}, we can use the UB bound of the channel error probability as a reliability metric to sort all the polarized channels. Consider the practical application, the logarithmic form of the UB bound is more convenient, which can be written as
\begin{equation}
\begin{aligned}
&\ln \left\{\sum \limits_{d} A_N^{(i)}(d) (Z(W))^d\right\}\\
&=\ln \left\{\sum \limits_{d} \exp \left[\ln {A_N^{(i)}(d)}+d\ln(Z(W))\right]\right\}\\
&\approx \max_{d}\left\{L_N^{(i)}(d)+d\ln (Z(W))\right\}.
\end{aligned}
\end{equation}
Here, we use the approximation $\ln\left( \sum\limits_{k} e^{a_k} \right)\approx \max\limits_{k}\{a_k\}$ and $L_N^{(i)}(d)=\ln\left[A_N^{(i)}(d)\right]$ is the logarithmic form of the polar weight enumerator.
Then given the B-DMC $W$ and the code length $N$, the reliability of the polarized channel can be sorted by the {\bf{UB Weight}} (UBW), that is,
\begin{equation}\label{UBW}
UBW_N^{(i)}=\max_{d}\left\{L_N^{(i)}(d)+d\ln (Z(W))\right\}.
\end{equation}


\begin{example}(UBW in BEC or BSC)
Given the Bhattacharyya parameter of the BEC $Z(W)=\epsilon$, UBW can be written as
\begin{equation}
UBW_N^{(i)}=\max_{d}\left\{L_N^{(i)}(d)+d\ln \epsilon\right\}.
\end{equation}
Similarly, UBW for the BSC $\left(Z(W)=2\sqrt {\delta(1-\delta)}\right)$ can be written as
\begin{equation}
UBW_N^{(i)}=\max_{d}\left\{L_N^{(i)}(d)+\frac{d}{2}\ln [4\delta(1-\delta)]\right\}.
\end{equation}
\end{example}

\begin{example}(UBW in AWGN channel)
For the BI-AWGN channel $W$, the Bhattacharyya parameter is $Z(W)=\exp(-\frac{E_s}{N_0})$. Thus the corresponding UBW can be expressed as
\begin{equation}\label{UBW_AWGN}
UBW_N^{(i)}=\max_{d}\left\{L_N^{(i)}(d)-d\frac{E_s}{N_0}\right\}.
\end{equation}
Especially, if the symbol SNR in UBW is instead of a fixed constant $\alpha$, we obtain a channel-independent construction metric.
\end{example}
\subsection{Construction Metrics based on simplified UB bound}
If we only consider the first term of the polar spectrum, that is, the minimum weight enumerator, the UBW given above can be further simplified as follows. 

Given the B-DMC $W$ and the code length $N$, the reliability of the polarized channel can be ordered by the {\bf{simplified UB Weight}} (SUBW), that is,
\begin{equation}
{SUBW}_N^{(i)}=L_N^{(i)}\left(d_{min}^{(i)}\right)+d_{min}^{(i)}\ln (Z(W)).
\end{equation}
Correspondingly, we can obtain the simplified union bound of the block error probability of polar codes
\begin{equation}
P_e(N,K,\mathcal{A})\lesssim \sum\limits_{i\in \mathcal{A}} A_N^{(i)}\left(d_{min}^{(i)}\right) P_N^{(i)}\left(d_{min}^{(i)}\right)
\end{equation}
or the simplified UB bound
\begin{equation}\label{equation65}
P_e(N,K,\mathcal{A})\lesssim \sum\limits_{i\in \mathcal{A}} A_N^{(i)}\left(d_{min}^{(i)}\right) (Z(W))^{d_{min}^{(i)}}.
\end{equation}

\begin{example}(SUBW in AWGN channel)
The SUBW for the BI-AWGN channel can be expressed as
\begin{equation}\label{SUBW_AWGN}
SUBW_N^{(i)}=L_N^{(i)}\left(d_{min}^{(i)}\right)-d_{min}^{(i)}\frac{E_s}{N_0}.
\end{equation}
Similarly, we can design a channel-independent construction metric by selecting a suitable constant $\alpha$ for the symbol SNR in SUBW.
\end{example}

\begin{example}
Table \ref{reliability_order} shows the reliability order of the polarized channels based on various constructions for the code length $N=32$ and the symbol SNR $\frac{E_s}{N_0}$ is fixed as $4$ dB. As shown in this table, the reliability order obtained by sorting UBW or SUBW from low to high is almost the same as that ordered by the GA algorithm except the polarized channels with indices 4 and 17. However, the reliability order obtained by sorting PW is different from that ordered by GA algorithm at the polarized channels with indices 20, 15, 8 and 25.
\end{example}

\begin{table*}[tp]
\centering
\caption{Reliability order example based on various constructions for $N=32$} \label{reliability_order}
\begin{tabular}{|c|c|c|c|c|}
\hline GA & PW &  Bhattacharyya & UBW & SUBW \\
\hline 32 & 32 & 32  & 32  & 32 \\
\hline 31 & 31 & 31  & 31  & 31 \\
\hline 30 & 30 & 30  & 30  & 30 \\
\hline 28 & 28 & 28  & 28  & 28 \\
\hline 24 & 24 & 24  & 24  & 24 \\
\hline 16 & 16 & 16  & 16  & 16 \\
\hline 29 & 29 & 29  & 29  & 29 \\
\hline 27 & 27 & 27  & 27  & 27 \\
\hline 26 & 26 & 26  & 26  & 26 \\
\hline 23 & 23 & 23  & 23  & 23 \\
\hline 22 & 22 & 22  & 22  & 22 \\
\hline 20 & 15 & 20  & 20  & 20 \\
\hline 15 & 20 & 15  & 15  & 15 \\
\hline 14 & 14 & 14  & 14  & 14 \\
\hline 12 & 12 & 12  & 12  & 12 \\
\hline 8  & 25 &  8  & 8   & 8  \\
\hline 25 & 8  & 25  & 25  & 25 \\
\hline 21 & 21 & 21  & 21  & 21 \\
\hline 19 & 19 & 19  & 19  & 19 \\
\hline 13 & 13 & 13  & 13  & 13 \\
\hline 18 & 18 & 18  & 18  & 18 \\
\hline 11 & 11 & 11  & 11  & 11 \\
\hline 10 & 10 & 10  & 10  & 10 \\
\hline 7  & 7  &  7  & 7   & 7  \\
\hline 6  & 6  &  6  & 6   & 6  \\
\hline 4  & 4  &  4  & 17  & 17 \\
\hline 17 & 17 & 17  & 4   & 4  \\
\hline 9  & 9  &  9 & 9   & 9  \\
\hline 5  & 5  &  5 &  5   & 5  \\
\hline 3  & 3  &  3 & 3   & 3  \\
\hline 2  & 2  &  2 & 2   & 2  \\
\hline 1  & 1  &  1 & 1   & 1  \\
\hline
\end{tabular}
\end{table*}

\begin{remark}
The UBW and SUBW are deduced from the error probability of the polarized channel, that is, union-Bhattacharyya bound analysis. They are determined by the polar spectrum and the Bhattacharyya parameter. Compared with the traditional constructions, such as, DE, GA, and Tal-Vardy algorithm etc., these metrics have explicitly analytical structure. Particularly, since the polar spectrum can be off-line calculated by using Algorithm \ref{algorithm1}, the complexity of these constructions is linear $O(N)$, which is much lower than that of the former algorithms. On the other hand, compared with the SNR-independent polarized weight (PW) construction, these constructions can also be modified to the SNR-independent constructions. However, PW is an empirical construction and lack of theoretic explanation. On the contrary, UBW or SUBW has a good analytical property. Therefore, the construction based on UBW or SUBW has dual advantages: theoretic analyticity and practical application.
\end{remark}
\section{Numerical Analysis and Simulation Results}
\label{section_VII}
In this section, we will provide the numerical and simulation results based on the polar spectrum. First, we compare the upper bounds in term of the polar spectrum with the traditional bounds based on  Tal-Vardy algorithm, GA or Bhattacharyya parameter under BEC, BSC and AWGN channel. Then, the BLER simulation results based on various constructions under AWGN channel are analyzed and compared.
\subsection{Numerical analysis of upper bounds}
In this part, we compare six different upper bounds of BLER which can be divided into two categories, the traditional bounds and the proposed ones based on the polar spectrum. The traditional upper bounds include the Tal-Vardy bound \cite{Tal_Vardy}, GA bound \cite{GA_Trifonov} and the bound provided by Ar{\i}kan in \cite{Polarcode_Arikan} as presented in (\ref{equation8}) while the proposed upper bounds contain the union bound given in (\ref{union_bound}), the UB bound expressed as (\ref{UB_bound}) and the simplified UB bound (\ref{equation65}).

\begin{figure*}[htbp]
\setlength{\abovecaptionskip}{0.cm}
\setlength{\belowcaptionskip}{-0.cm}
  \centering{\includegraphics[scale=0.93]{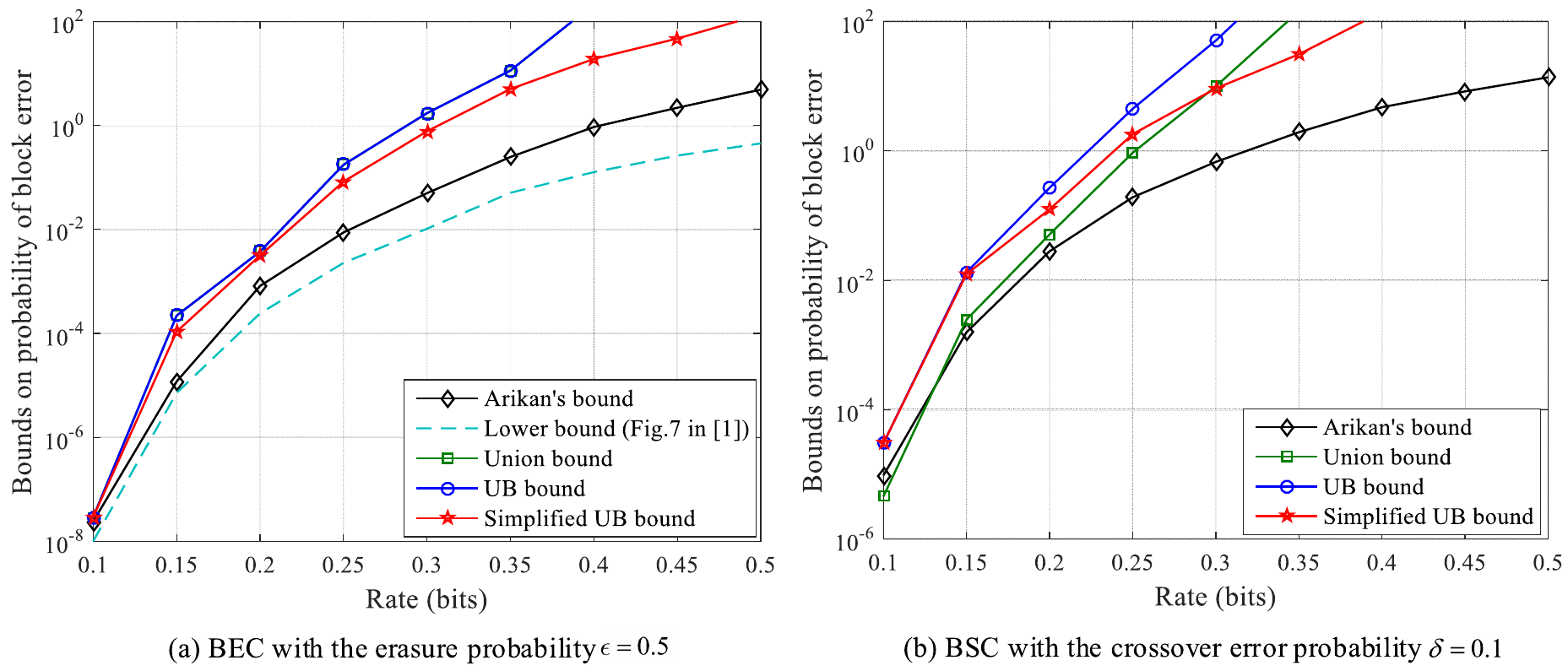}}
  \caption{The BLER upper bounds for polar coding with SC decoding at code length $N=128$ under different channel conditions.}\label{upper_bound_BEC_BSC}
\end{figure*}

Given the code length $N = 128$, Fig. \ref{upper_bound_BEC_BSC} provides the upper bounds of BLER using SC decoding under BEC with the erasure probability $\epsilon=0.5$ and BSC with the crossover error probability $\delta=0.1$. The upper bound presented in (\ref{equation8}) is named as ``Ar{\i}kan's bound'' hereafter and the lower bound marked by the dash line is obtained from $\max\limits_{i\in\mathcal{A}}\left\{Z\left(W_N^{(i)}\right)\right\}$ (See Fig. 7 in \cite{Polarcode_Arikan}). As shown in Fig. \ref{upper_bound_BEC_BSC}(a), all the upper bounds grow larger as the code rate increases. However, it can be observed that the UB bound and the union bound (\ref{equation23}) coincide under the BEC and they tend to diverge when the code rate is greater than $0.35$, while the simplified UB bound avoids this problem and is looser than the Ar{\i}kan's bound. When the BSC is considered, as shown in Fig. \ref{upper_bound_BEC_BSC}(b), the union bound (\ref{equation26}) is closer to the Ar{\i}kan's bound when the code rate is smaller than $0.3$, otherwise, the simplified UB bound is closer.

\begin{figure*}[htbp]
\setlength{\abovecaptionskip}{0.cm}
\setlength{\belowcaptionskip}{-0.cm}
  \centering{\includegraphics[scale=0.91]{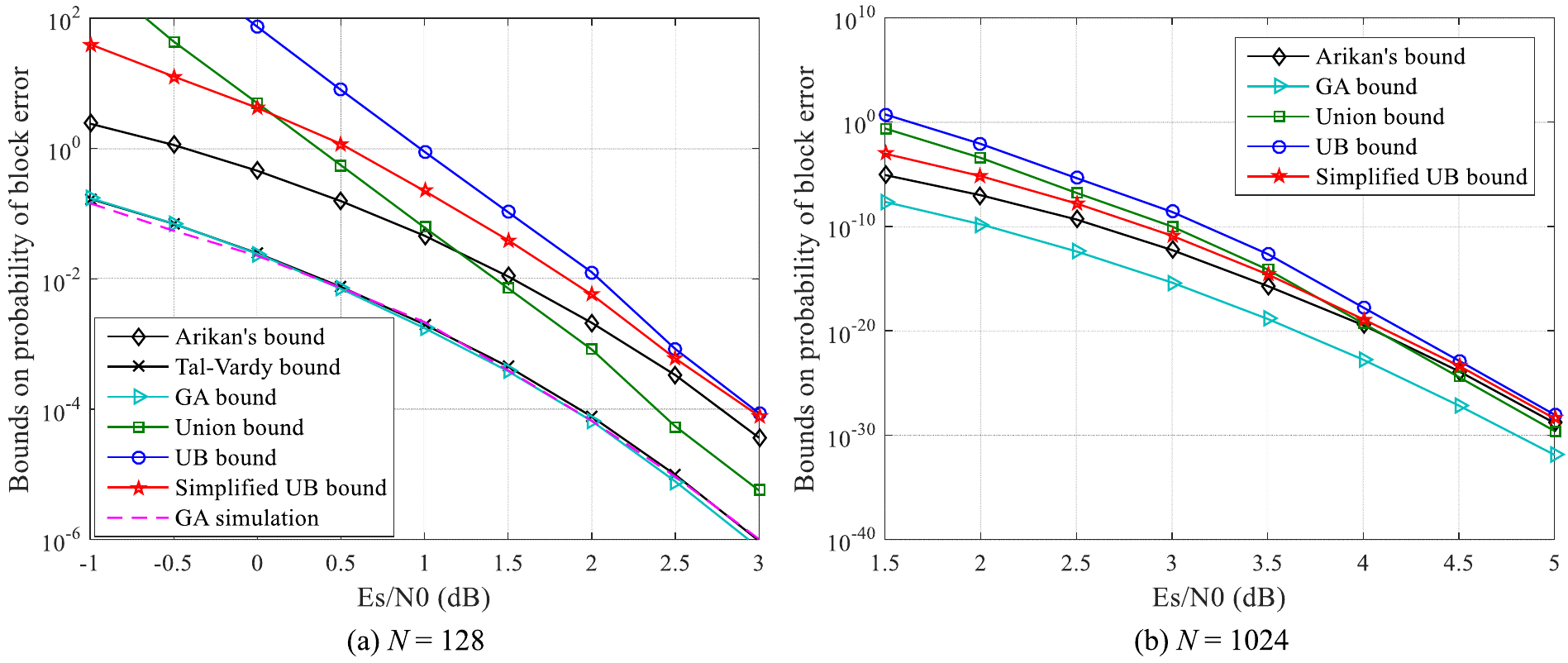}}
  \caption{The BLER upper bounds for polar coding with SC decoding over AWGN channel, where the code rate $R=0.5$.}\label{upper_bound_AWGN_N128_1024}
\end{figure*}

The BLER upper bounds of SC decoding under AWGN channel are depicted in Fig. \ref{upper_bound_AWGN_N128_1024}. As shown in Fig. \ref{upper_bound_AWGN_N128_1024}(a), where the code length $N$ is set to $128$ and the code rate is $R=0.5$, all the upper bounds dramatically decrease with the increase of the symbol SNR. Due to the approximation calculation of Bhattacharyya parameter in AWGN channel, Ar{\i}kan's bound is looser than GA and Tal-Vardy bounds. In addition, the union bound (\ref{equation33}) tends to be more closer to the simulation result (marked by dash line and constructed based on GA algorithm) than the Ar{\i}kan's bound when the symbol SNR exceeds $1.2$ dB. It also can be observed that the UB bound (\ref{equation40}) and the simplified UB bound (\ref{equation65}) are gradually close to the Ar{\i}kan's bound with the increase of the symbol SNR. Similar observations can be found in Fig. \ref{upper_bound_AWGN_N128_1024}(b), which provides the BLER upper bounds of SC decoding under AWGN channel with the configuration $N=1024$ and $R=0.5$.

Although the Tal-Vardy and GA bounds are tightly close to the BLER simulation result, they involve a complex on-line iterative-calculation. On the contrary, the proposed UB and simplified UB bounds not only have a linear complexity since the polar spectrum can be calculated off-line, but also can be used to deduce two explicit and analytical construction metrics to construct good polar codes. The BLER simulation results based on these explicit constructions will be presented in the next subsection.

\subsection{Simulation Results}
In this part, considering the AWGN channel, we compare the BLER simulation performances of polar codes generated by the proposed UBW/SUBW constructions and the traditional methods, which include Tal and Vardy's algorithm \cite{Tal_Vardy}, GA \cite{GA_Trifonov}, PW \cite{PW_He} and the one based on Bhattacharyya parameter proposed by Ar{\i}kan in \cite{Polarcode_Arikan}. The code length $N$ is in \{128, 1024, 4096\} and the code rate $R$ is choose from \{$\frac{1}{3}$, $\frac{1}{2}$, $\frac{2}{3}$\}. For Tal and Vardy's method, the output alphabet size $\mu$ is set to $256$.
\begin{figure*}[htbp]
\setlength{\abovecaptionskip}{0.cm}
\setlength{\belowcaptionskip}{-0.cm}
  \centering{\includegraphics[scale=0.95]{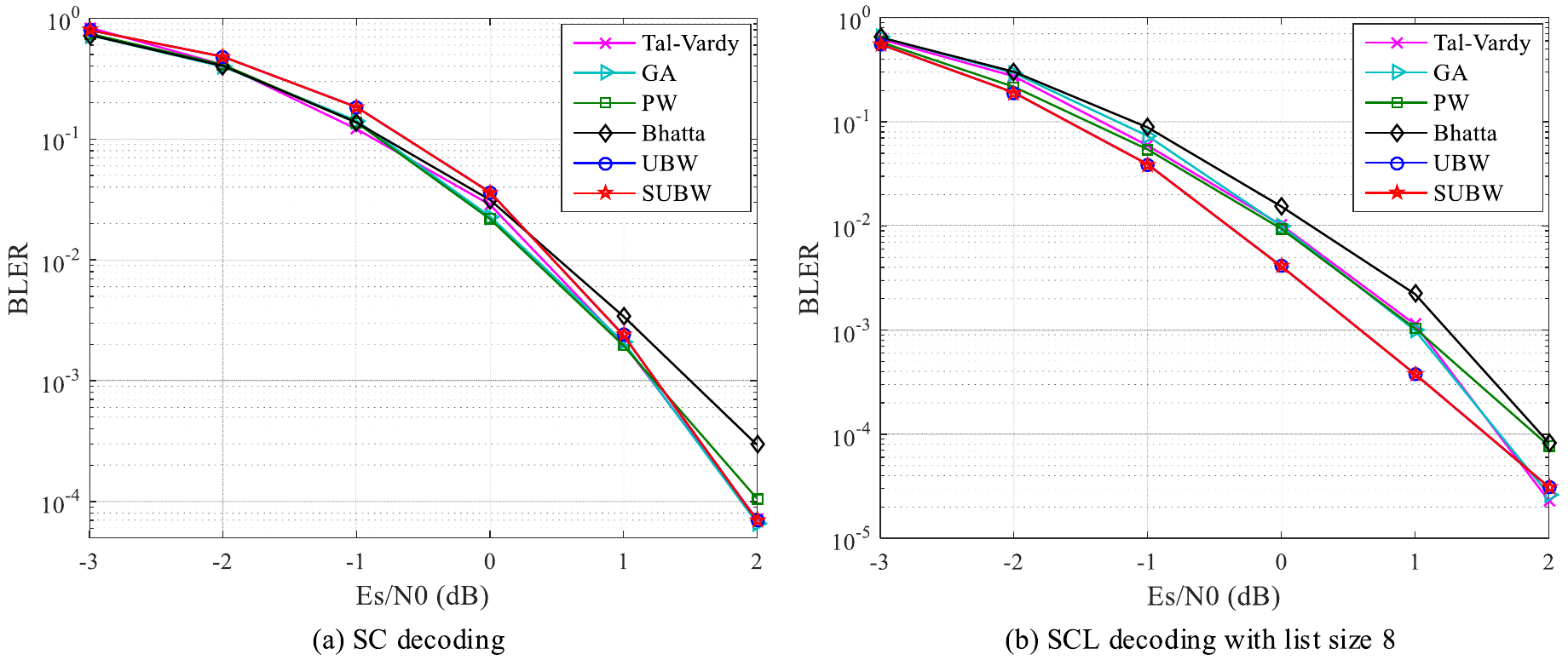}}
  \caption{The BLER performances comparison among the polar codes constructed based on Tal and Vardy's algorithm, GA, PW, Bhattacharyya parameter, UBW and SUBW under AWGN channel, where $N = 128$ and $R=0.5$.}\label{AWGN_N128_SC_SCL}
\end{figure*}

Fig. \ref{AWGN_N128_SC_SCL} provides the BLER performances comparison among the various constructions with $N = 128$ and $R=0.5$. For UBW/SUBW, the fixed symbol SNR in (\ref{UBW_AWGN})/(\ref{SUBW_AWGN}) is set to $\alpha=4$ dB. It can be observed from Fig. \ref{AWGN_N128_SC_SCL}(a) that the polar codes constructed by UBW/SUBW can achieve similar performance of those constructed by Tal-Vardy, GA or PW algorithm under SC decoding. Moreover, the polar codes constructed by UBW/SUBW outperform those constructed based on the Bhattacharyya parameter in the case of $\frac{E_s}{N_0}>0.28 \text{dB}$. When the SCL decoding with list size $L=8$ is used, as shown in Fig. \ref{AWGN_N128_SC_SCL}(b), the polar codes constructed by UBW/SUBW outperform those generated by the traditional constructions by about $0.4\sim0.6$ dB as the BLER is $10^{-3}$.

\begin{figure}[htbp]
\setlength{\abovecaptionskip}{0.cm}
\setlength{\belowcaptionskip}{-0.cm}
  \centering{\includegraphics[scale=0.7]{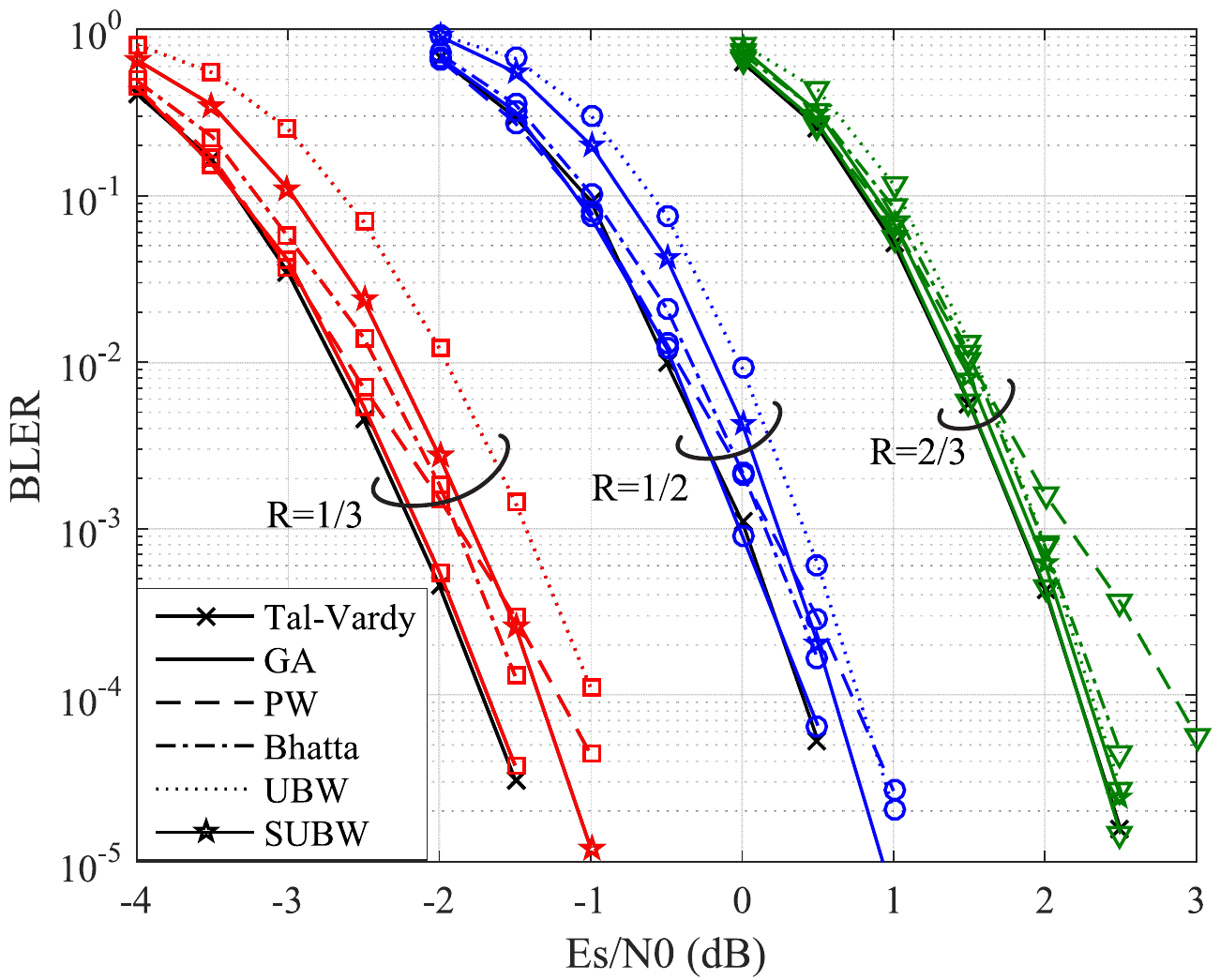}}
  \caption{The BLER performances comparison among the polar codes constructed based on Tal and Vardy's algorithm, GA, PW, Bhattacharyya parameter, UBW and SUBW under AWGN channel, where $N = 1024$ and SC decoding is used.}\label{AWGN_N1024_SC}
\end{figure}

Given the code length $N=1024$, Fig. \ref{AWGN_N1024_SC} presents the comprehensive BLER performances comparison among the various constructions under SC decoding. The symbol SNR in (\ref{UBW_AWGN}) for UBW is respectively set to $1.5$, $4$ and $4.5$ dB for the code rate $R=\frac{1}{3}$, $\frac{1}{2}$ and $\frac{2}{3}$, while that for SUBW in (\ref{SUBW_AWGN}) is respectively fixed as $1$, $3.5$ and $4$ dB. As shown in Fig. \ref{AWGN_N1024_SC}, when a low to medium code rate is considered, namely, $R=\frac{1}{3}$  or $\frac{1}{2}$, the polar codes constructed by UBW/SUBW perform slightly worse than those constructed by GA or Bhattacharyya parameter because the UB bound of the BLER under the SC decoding is looser than the GA bound or Ar{\i}kan's bound. However, similar to the PW construction, the UBW/SUBW constructions ensure the explicity at the cost of some performance losses.
In addition, when the code rate $R=\frac{2}{3}$ is considered, the polar codes constructed by UBW/SUBW can achieve nearly the same performance of those constructed by Tal-Vardy or GA. Apart from these, one can also observe that the polar codes constructed by PW show error floor in the high SNR region, which can be avoided by UBW/SUBW.

\begin{figure}[htbp]
\setlength{\abovecaptionskip}{0.cm}
\setlength{\belowcaptionskip}{-0.cm}
  \centering{\includegraphics[scale=0.7]{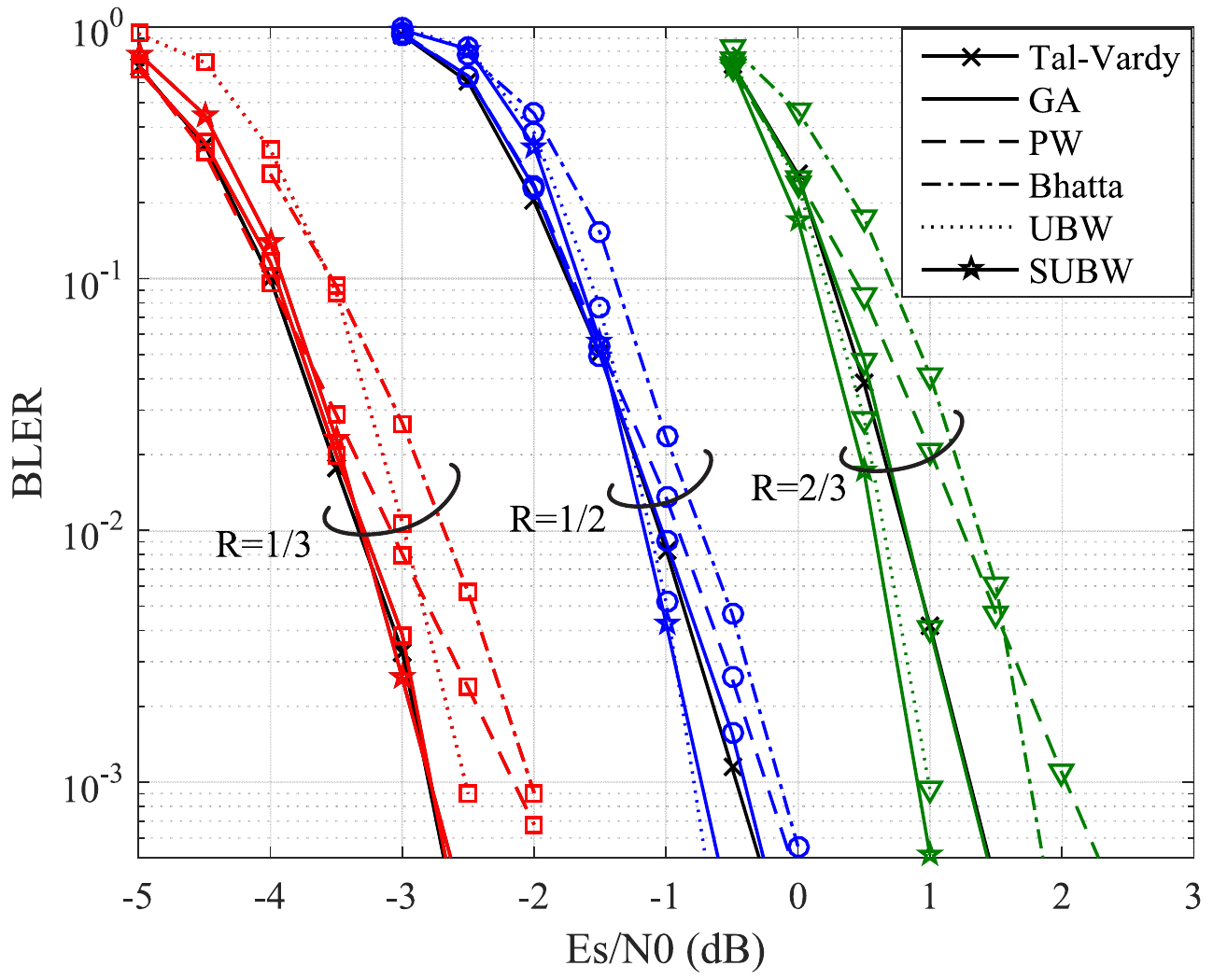}}
  \caption{The BLER performances comparison among the polar codes constructed based on Tal and Vardy's algorithm, GA, PW, Bhattacharyya parameter, UBW and SUBW under AWGN channel, where $N = 1024$ and SCL decoding with list size 16 is used.}\label{AWGN_N1024_SCL}
\end{figure}

For the code length $N=1024$, the BLER performances comparison among the various constructions under SCL decoding with list size $L=16$ is shown in Fig. \ref{AWGN_N1024_SCL}. It can be observed that the polar codes constructed by UBW/SUBW can achieve better performance than those constructed by Tal-Vardy/GA/PW or the construction based on Bhattacharyya parameter. The performance gain of polar codes constructed by UBW/SUBW becomes larger with the increase of code rate. For example, given $R=\frac{1}{2}$, when the BLER is $10^{-3}$, polar codes constructed by SUBW achieve $0.3$/$0.5$ dB gain compared to those constructed by GA/PW, it then becomes $0.4$/$1.3$ dB under $R=\frac{2}{3}$. Actually, polar codes constructed by UBW/SUBW benefit from the adequate utilization of the polar spectrum, which is vital for the polar code construction under the SCL decoding. Furthermore, polar codes constructed by SUBW outperform those constructed by UBW when $R=\frac{1}{3}$ because the minimum weight term of the polar spectrum plays an important role in the analysis of error probability of the polarized channels.

\begin{figure*}[htbp]
\setlength{\abovecaptionskip}{0.cm}
\setlength{\belowcaptionskip}{-0.cm}
  \centering{\includegraphics[scale=0.93]{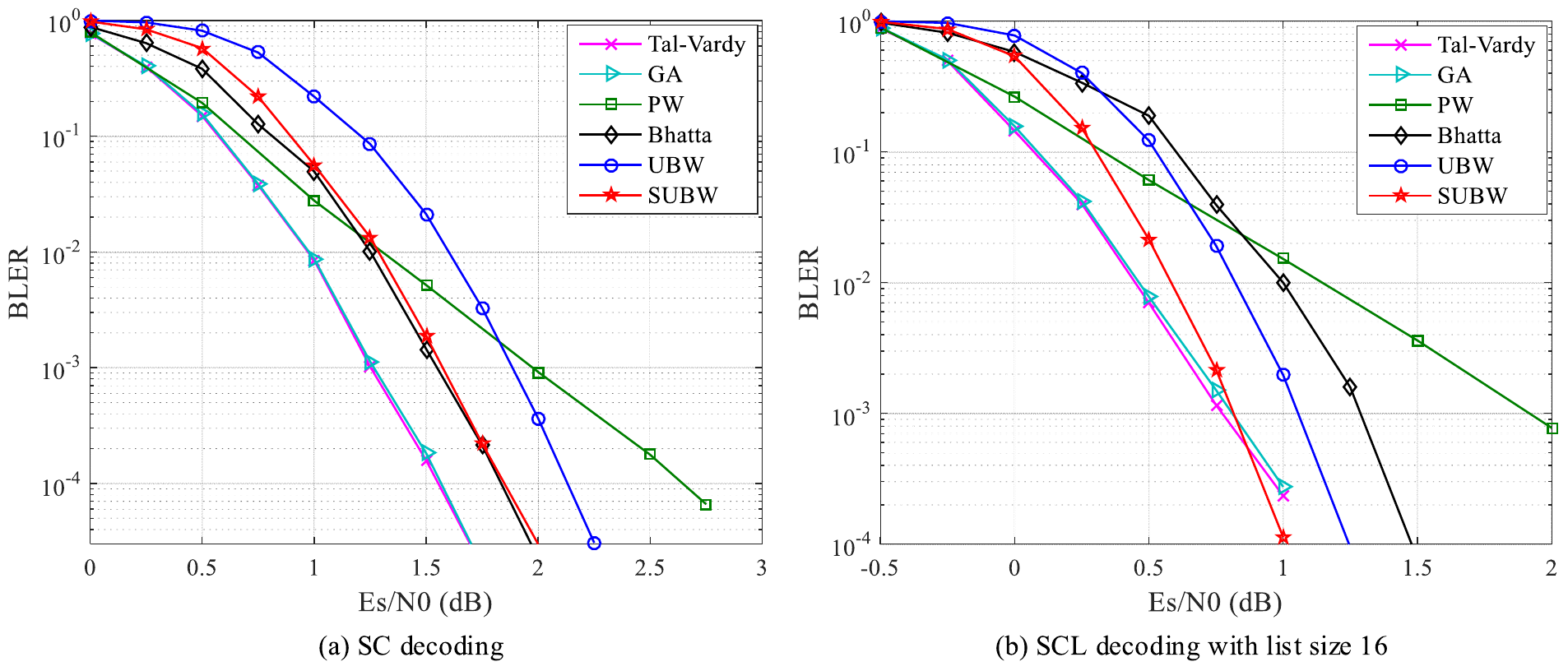}}
  \caption{The BLER performances comparison among the polar codes constructed based on Tal and Vardy's algorithm, GA, PW, Bhattacharyya parameter, UBW and SUBW under AWGN channel, where $N = 4096$ and $R=\frac{2}{3}$.}\label{AWGN_N4096_SC_SCL}
\end{figure*}

Fig. \ref{AWGN_N4096_SC_SCL} provides the BLER performances comparison among the various constructions with $N = 4096$ and $R=\frac{2}{3}$. In this case, the symbol SNR in (\ref{UBW_AWGN})/(\ref{SUBW_AWGN}) for UBW/SUBW is set to $4.5$/$3.5$ dB. It can be observed from Fig. \ref{AWGN_N4096_SC_SCL}(a) that polar codes constructed by UBW/SUBW perform worse than those constructed by Tal-Vardy or GA under SC decoding but avoid the error floor in the high SNR region, which is obviously shown in the polar codes constructed by PW. When the SCL decoding with list size $L=16$ is used, as shown in Fig. \ref{AWGN_N4096_SC_SCL}(b), the polar codes constructed by SUBW can achieve $0.5$/$1.1$ dB gain compared to those constructed by Bhattacharyya parameter/PW when the BLER is $10^{-3}$.

In general, UBW and SUBW can be regarded as two good constructions, which are simple and explicit for the practical polar coding as well as able to generate polar codes with superior performance under SCL decoding over those based on the traditional methods in most cases.
\section{Conclusions}
\label{section_VIII}
In this paper, we introduce a new concept named polar spectrum from the codeword weight distribution of polar codes and then establish a systematic framework in term of the polar spectrum to analyze and construct polar codes. On the basis of the polar spectrum, we derive the union bound and UB bound of the polarized channel and further upper bound the BLER of SC decoding. In addition, we propose an iterative algorithm embedded the solution of MacWilliams identities to enumerate the polar spectrum of polar codes, which has a low complexity and a high efficiency compared with the traditional searching algorithm. Finally, we design two explicit and analytical construction metrics named UBW and SUBW, which have a linear complexity far below those constructions based on the iterative calculation since the polar spectrum can be calculated off-line. Simulation results show that the polar codes constructed by these two metrics can achieve similar performance of those constructed by GA or PW algorithm under SC decoding and even superior performance under SCL decoding.

\end{document}